\documentclass{article}

\usepackage[english]{babel}

\usepackage{booktabs}
\usepackage{url}
\usepackage{cite}
\usepackage{amsmath}
\usepackage{amssymb}
\usepackage{verbatim}
\usepackage{graphicx}
\usepackage{rotating}
\usepackage{microtype}
\usepackage{color}
\usepackage{xspace}
\usepackage[caption=false]{subfig}
\usepackage[noend]{algpseudocode}

\usepackage{caption} 
\usepackage[bm,compactparagraph,algpseudocodecmds]{pauli_new}

\renewcommand{\rem}[1]{\textbf{{\color{red} #1}}}
\renewcommand{\rem}[1]{} 

\graphicspath{{figs_embeded/}}

\newlength{\singlefigwidth}
\newcommand{\modelname}[1]{\texttt{#1}}
\newcommand{\mhyperbolic}[1]{\modelname{hyperbolic}\ifx&#1&%
\xspace
\else
(#1)\xspace
\fi}
\newcommand{\mmixture}[1]{\modelname{mixture}\ifx&#1&%
\xspace
\else%
(#1)\xspace
\fi}
\newcommand{\mfixed}[1]{\modelname{fixed}\ifx&#1&%
\xspace
\else%
(#1)\xspace
\fi}

\renewcommand{\complement}[1]{\mkern 1.5mu\overline{\mkern-1.5mu#1\mkern-1.5mu}\mkern 1.5mu}

\newcommand{\youtube}{\datasetname{YouTube}\xspace}
\newcommand{\amazon}{\datasetname{Amazon}\xspace}
\newcommand{\friendster}{\datasetname{Friendster}\xspace}
\newcommand{\lj}{\datasetname{LiveJournal}\xspace}
\newcommand{\orkut}{\datasetname{Orkut}\xspace}
\newcommand{\dblp}{\datasetname{DBLP}\xspace}

\newcommand{\polbooks}{\datasetname{PolBooks}\xspace}
\newcommand{\jazz}{\datasetname{Jazz}\xspace}
\newcommand{\emaild}{\datasetname{Email}\xspace}
\newcommand{\erdos}{\datasetname{Erd\H{o}s}\xspace}

\title{Hyperbolae Are No Hyperbole:\\ Modelling Communities That Are Not Cliques}

\author{
Saskia Metzler\\
\normalsize{Max Planck Institute for Informatics}\\
\normalsize{Saarland Informatics Campus, Germany}\\
\normalsize{\texttt{smetzler@mpi-inf.mpg.de}}
\and
Stephan G\"unnemann\\
\normalsize{Dept. of Informatics \& Inst. for Advanced Study}\\
\normalsize{Technical University of Munich, Germany}\\
\normalsize{\texttt{guennemann@in.tum.de}}
\and
Pauli Miettinen\\
\normalsize{Max Planck Institute for Informatics}\\
\normalsize{Saarland Informatics Campus, Germany} \\
\normalsize{\texttt{pmiettin@mpi-inf.mpg.de}}
}

\date{}
\begin{document}
	
\maketitle
\begin{abstract}
Cliques  are frequently used to model  communities: a community is a set of nodes where each pair is equally likely to be connected. But studying real-world communities reveals that they have more structure than that. In particular, the nodes can be ordered in such a way that (almost) all edges in the community lie below a hyperbola. In this paper we present three new models for communities that capture this phenomenon. Our models explain the structure of the communities differently, but we also prove that they are identical in their expressive power. Our models fit to real-world data much better than traditional block models or previously-proposed hyperbolic models, both of which are a special case of our model. Our models also allow for intuitive interpretation of the parameters, enabling us to summarize the shapes of the communities in graphs effectively.


\end{abstract}


\section{Introduction}
\label{sec:introduction}

Community detection in graphs has gathered significant research interest in recent years. So far, most approaches have (explicitly or implicitly) modelled communities as (quasi\nobreakdash-) cliques, i.e.\ sets of nodes where every node is connected to (almost) every other node, that is, every edge within the community is equally likely. 

We argue that a clique is often not a realistic model for real-world communities. Consider the adjacency matrix of a community from the \youtube data from the Stanford Large Network Dataset collection~\cite{leskovec14snap} in Figure~\ref{fig:exampleCommYoutube}. In its original ordering, the community does look like a quasi-clique. But if we order the nodes by their degree (Figure~\ref{fig:exampleCommYoutube_ordered}), we see immediately that the  community takes the shape of a hyperbola: there is a clear curve such that it is much more likely to see an edge below this curve than above it. Hence, \emph{not every edge between a pair of nodes in a community is equally likely}; 
a phenomenon that has been observed in multiple real-world graphs~\cite{araujo2014beyond}.

While \cite{araujo2014beyond} introduces an initial model to address this aspect, the proposed model is very restricted, capturing only certain shapes of communities and ignoring the inherent sparsity of real world networks (see Sections~\ref{sec:related-work} and \ref{sec:hycom} for a detailed discussion). To overcome these limitations, we propose a novel model for  communities that explicitly captures the shape of the edge distribution and that incorporates sparsity in its probabilistic formulation. Most importantly, our model contains (quasi-) cliques as well as the model of~\cite{araujo2014beyond} as a special case.

\begin{figure}[tb]
	\vspace*{-1mm}
	\centering
	\subfloat[Unordered.]{
		\centering
		\includegraphics[width=0.472 \linewidth]{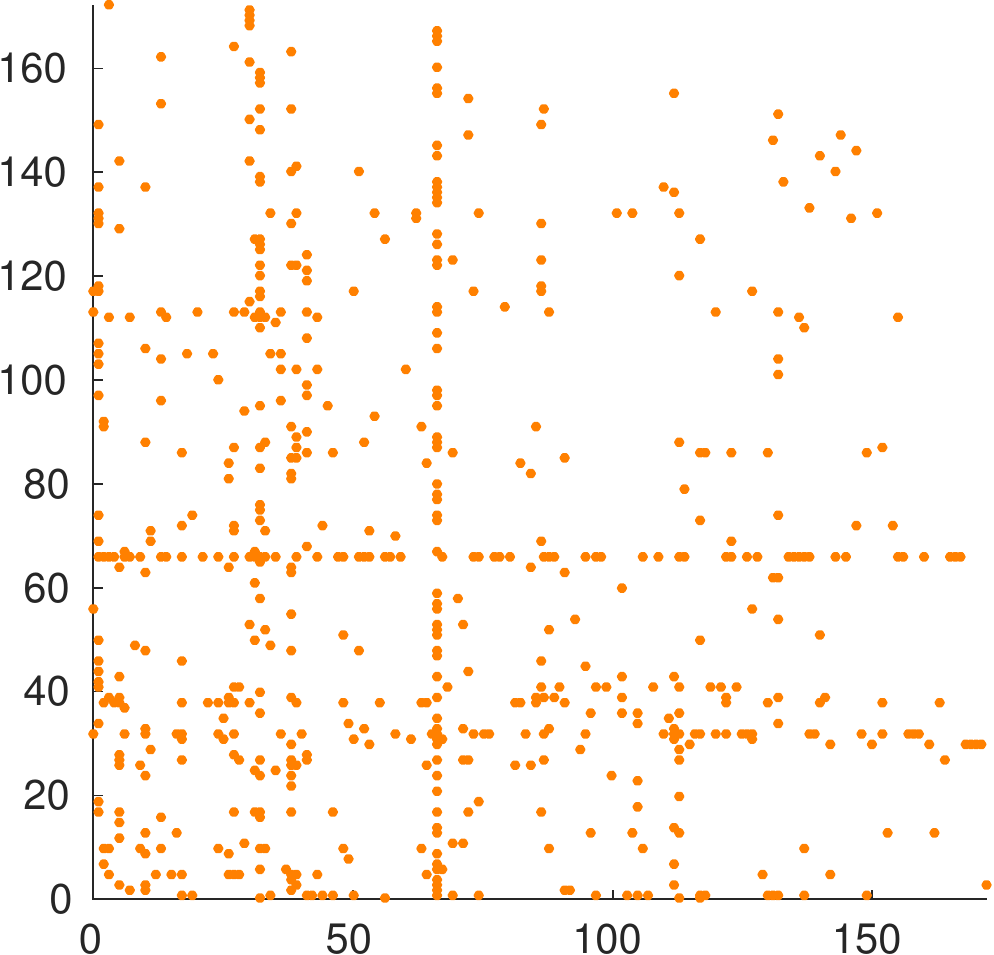}
		\label{fig:exampleCommYoutube}
	}%
	\subfloat[Degree ordered.]{
		\centering
		\includegraphics[width=0.472 \linewidth]{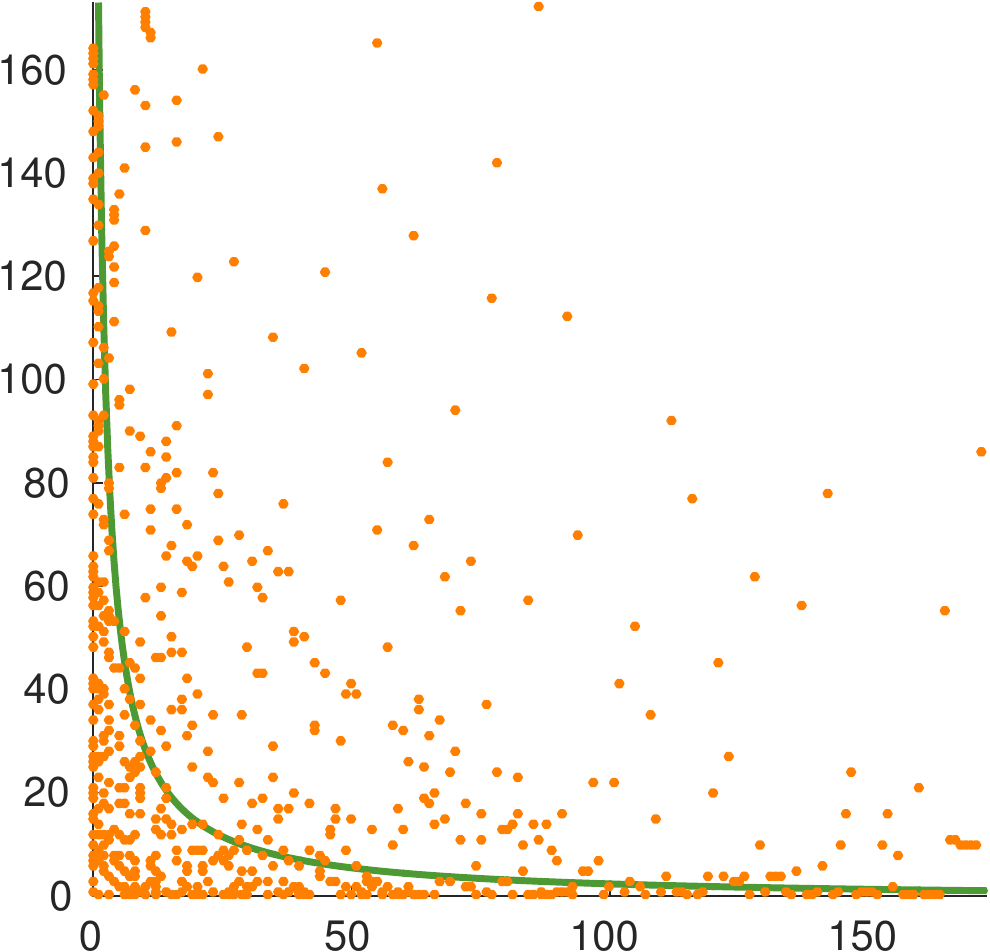}
		\label{fig:exampleCommYoutube_ordered}
	}
	\caption{Adjacency matrix of a community from the \youtube data unordered and ordered by induced degree with a model fitted.}
	\label{fig:exampleCommYoutube_both}
\end{figure}

In fact, we present three different models that allow us to capture different features of the communities and apply different optimization techniques when fitting the model in Section~\ref{sec:models}. We will also prove that these three models are equivalent in the sense that a community in one model can be easily transformed to an equivalent community in another model. Yet, the models are not redundant, as the different views they provide allow for more in-depth understanding of the communities (and the models, as well) -- and they show how our approach generalizes existing models.

 In addition to modelling individual communities, we also present a model for the full graph given a set of communities in Section~\ref{sec:full-graph}. In Section~\ref{sec:complexity} we study the computational complexity of some problems related to the modelling and in Section~\ref{sec:algorithms} we present our algorithms for fitting the models to individual communities and the full graph alike. 

Our experimental evaluation, in Section~\ref{sec:experiments}, shows that we can efficiently model real-world graphs, yielding significantly more likely models than what we can get by modelling the communities as quasi-cliques or existing hyperbolic models. We will also demonstrate how our models will allow us to gain more understanding on the structure of the communities.

Our main contributions in this paper are
\begin{itemize}
\item We present three different but equivalent models for capturing non-uniform edge distributions inside communities.
\item We show how to fit the models to the communities, and how to fit a model for the full graph to a graph with a set of communities.
\item We show that our approach explains real graphs much better than traditional quasi-clique based models and also improves the model of~\cite{araujo2014beyond}.
\end{itemize}




\section{Related Work}
\label{sec:related-work}

Nodes in real-world networks -- including social networks and gene-regulatory networks -- often organize into communities or clusters. 
A variety of methods have been introduced in the literature for finding the communities~\cite{aggarwal2010managing}.
Even though the proposed approaches seem highly diverse, the vast amount of previous community detection methods have been either explicitly or implicitly aimed at detecting \emph{block-shaped} areas of \emph{uniform density} in the adjacency matrix. This includes prominent techniques such as stochastic block-models \cite{nowicki2001estimation,airoldi2009mixed}, affiliation network models \cite{leskovec:agmfit}, pattern based techniques (e.g.\ detection of quasi-cliques \cite{DBLP:journals/kais/GunnemannFBS14,DBLP:conf/kdd/BodenGHS12}), or cross-associations~\cite{crossassocs}.

 In this paper, following the observations of \cite{araujo2014beyond}, we argue that communities in real networks do not show such a density profile. As illustrated in Figure~\ref{fig:exampleCommYoutube_both}, they are better represented by using a hyperbolic model. Our model captures that the density within the group is not uniformly distributed -- some nodes show stronger connectivity among its peers.

 \paragraph{Overlapping community detection} Unequal connectivity structure has also been considered in the area of overlapping community detection. Here, in contrast to classical partitioning approaches, each node might belong to multiple groups. As, e.g., noticed in \cite{leskovec:agmfit}, nodes participating in multiple communities lead to areas of higher density. Accordingly, one might conclude that these approaches can handle our scenario of hyperbolic communities. This conclusion, however, is incorrect.
 
 Given a certain community (i.e.\ a set of nodes), models following the idea of affiliation networks \cite{leskovec:agmfit}, generate edges following the same probability; thus, leading to block-like structures in the adjacency matrix. 
 Similarly, techniques such as mixed-membership block models \cite{airoldi2009mixed} assume uniform density per community.
 Indeed, all these methods are highly related to the principle of Boolean Matrix factorization \cite{miettinen2008discrete} -- and its goal of finding overlapping \emph{blocks} in binary matrices. 
 Clearly, handling overlap and handling hyperbolic structure are two different aspects. As we will discuss in Section \ref{sec:full-graph}, for hyperbolic communities itself, three different types of (non-) overlap can be considered.
 
  \paragraph{Hyperbolic community detection}
The aspect of non-block communities has been discussed in \cite{araujo2014beyond}.
The proposed model, called HyCom, still does not capture real scenarios well due to two reasons: First, the modeled hyperbolic communities are restricted in their possible shapes. That is, not all patterns appearing in real data can be well represented by the model.
Indeed, as we will show formally in Section~\ref{sec:hycom}, our model contains \cite{araujo2014beyond} as a special case. Second, \cite{araujo2014beyond}  assumes a density of 100\% inside a community, thus violating the general property of sparsity. In contrast, our model allows varying density among the different communities. The benefits of our model are also clearly confirmed in the experimental analysis. Nested matrices, of which hyperbolic graphs are a special case, have nonnegative rounding rank 1~\cite{neumann16what}, suggesting that rounding rank decompositions could provide an approach for finding hyperbolic communities. Yet, the connection to nested matrices does not generalize to higher rounding ranks. 

  \paragraph{Other graph patterns}
  For the purpose of graph compression, other graph patterns going beyond quasi-cliques have been considered. In \cite{DBLP:journals/tkde/LimKF14}, graphs are considered as a collection of hubs connecting to spokes. These hubs are recursively connected to super-hubs and so on. Extending this idea, the work \cite{DBLP:journals/sadm/KoutraKVF15} compresses graphs by using patterns such as stars or bipartite cores.
  None of these works exploits the idea of hyperbolic community structure.


\section{Models for Communities}
\label{sec:models}
Our goal is to model the aforementioned structure in communities and to that end we present three different models. It should be noted that all these models contain the (quasi-) cliques and the model of \cite{araujo2014beyond} as special cases (see Section~\ref{sec:hycom}). We will also show that these three different models -- with different intuitions behind them -- are equivalent representations besides having different parameterizations. These different models not only give three different interpretations but also enable ways to easily compare with other approaches and to efficiently compute their fit to given data.

\subsection{General modelling decisions}
\label{sec:model:general}

Our input is an undirected graph $G = (V, E)$ with $n$ nodes and $m$ edges. We will assign a number from $\{0, \ldots, n-1\}$ to the vertices and use $(i,j)$ to denote both a pair of vertices and the (potential) undirected edge between them. We will represent the graph using its \emph{adjacency matrix} $\mA = (a_{ij}) \in \B^{n\times n}$.

A \emph{community} $C$ is a tuple $(V_C, \pi_C, \Theta_C)$. The set $V_C\subseteq V$ contains the nodes in the community, and we write $n_C = \abs{V_C}$. The permutation $\pi:V_C\to \{0, \ldots, n_C - 1\}$ orders the nodes. In general, we assume the nodes to be ordered according to their degrees inside the community. The crucial part of our model is the following: not every edge between the nodes in $V_C$ is necessarily part of our community -- that assumption would make all of our communities quasi-cliques. Rather, our community models are defined using functions $f: \{0,\ldots, n_C - 1\} \times \{0,\ldots, n_C -1\}\times \Theta \to \{0, 1\}$ operating on a set of parameters $\Theta$ and deciding for any pair of vertices $(i, j)\in \{0,\ldots,n_C-1\}\times \{0, \ldots, n_C-1\}$ if an edge between $i$ and $j$ is part of the community or not. We will define these functions in the subsequent sections. 

Notice that the function $f$ only gets the indices relative to the subgraph, not to the full graph, that is, to test a pair $(i, j)\in V_C\times V_C$, we need to compute $f(\pi_C(i), \pi_C(j), \Theta_C)$. For brevity, we will often omit the permutation and will simply assume that $i,j\in\{0, \ldots, n_C - 1\}$.

Every community is associated with two sets of edges: the \emph{area} of the community, $A_C$, defined as $A_C = \{(i, j)\in V_C\times V_C : f(i, j, \Theta_C) = 1\}$, and the \emph{edges} of it, $E_C = E \cap A_C$. For notational convenience, we also define their complements (with respect to the community and the area, respectively): $\complement{A_C} = (V_C\times V_C) \setminus A_C$ and $\complement{E_C} = A_C\setminus E$.

\rem{Where do we explain that the final order is based on the model, not on the actual edges?}

\paragraph{Probabilistic model for a community}
In practice the communities are rarely, if ever, exact. That is, some edges $(i,j)\in A_C$ are not in $E_C$, and some edges $(i,j)\in E$ that go between the nodes of the community are not in $A_C$. To model these imperfect communities, we consider a  \emph{probabilistic model} of the community. Given a community $C$, we assume that edges $(i,j)\in V_C\times V_C$ are drawn from a Bernoulli distribution, $a_{i,j} \sim \text{Bernoulli}(p_*)$, where $\mA = (a_{ij})$ is the adjacency matrix of the graph, and $p_*$ is the \emph{density} of the area that the edge belongs to. For a single community, we have two kinds of areas: the area of the community $A_C$ and its complement $\complement{A_C}$. 
We denote the density of the area of the community by 
\begin{equation}
\label{eq:d_C}
d_C = \abs{E_C}/\abs{A_C} \; , 
\end{equation}
and the density of the area outside the community by 
\begin{equation}
\label{eq:d_O}
d_O = \abs*{E\cap \complement{A_C}}/\abs{\complement{A_C}} \; .
\end{equation}

These densities correspond to the maximum-likelihood solutions of the variables $p_*$ for the edges that are inside or outside of the community. We can now consider the likelihood of the subgraph induced by the community $G|_{V_C}$ given the community $C$, $L(G|_{V_C} \mid C)$. The likelihood of an edge that is in community $C$ is $d_C$; the likelihood of a pair $(i,j)$ that is in the area of $C$ but that is not an edge of $G$ is $1-d_C$; the likelihood of an edge that is not in the community is $d_O$; and the likelihood of a pair $(i ,j)$ that is not in the community's area and is not an edge of $G$ is $1 - d_O$. This gives us
\begin{equation}
  \label{eq:comm_LL}
  \begin{split}
  \log L(G|_{V_C} \mid C) &= \abs{E_C}\log(d_C) + \abs*{\complement{E_C}}\log(1-d_C) \\
  &\quad + \abs*{E\cap \complement{A_C}}\log(d_O) \\
  &\quad + \abs*{\complement{A_C} \setminus E}\log(1 - d_O)\; .
  \end{split}
\end{equation}

\subsection{Area restricted under a hyperbola}
\label{sec:model:hyperbola}

For our first model, we can notice from Figure~\ref{fig:exampleCommYoutube_both} that when the nodes of a community are ordered in the induced degree order, the edges lie under a hyperbolic curve. To define the hyperbola we identify the vertex indices of the community as points in $x$ and $y$ axes. We use $i$ and $j$ instead of $x$ and $y$ to emphasize this connection. We will only consider the area $[0, n_C -1]\times [0,n_C -1]$ from the non-negative quadrant, as that is where the values important to our community are. The equation for a (rectangular) hyperbola is
\begin{equation}
  \label{eq:hyperbola}
  (i + p)(j + p) = \theta\; ,
\end{equation}
with the centre at $(-p, -p)$. Figure~\ref{fig:hyperbola} illustrates one hyperbola. Following the model, an edge $(i, j)$ is considered to be in the community if $(i + p)(j + p) \leq \theta$. We call this model \mhyperbolic{$p, \theta$} and write $(i, j) \in \mhyperbolic{p, \theta}$ if $(i + p)(j + p) \leq \theta$.

From Figure~\ref{fig:hyperbola} we can gain some intuition to the parameters $p$ and $\theta$: different values of $p$ will yield different shapes of the gradient (the coloured background in Figure~\ref{fig:hyperbola}), while different values of $\theta$ will move the line away from the origin.

\paragraph{Valid range of parameters} The values \eqref{eq:hyperbola} assigns to elements $(i, j)$ attain their minimum at the centre $(-p, -p)$. To make sure that all elements $(i, j)\in\N\times\N$ that are under the curve \eqref{eq:hyperbola} are always in the community, we must bound $p$. A simple boundary is to enforce that $p \geq 0$, though in the next section  we will derive a more relaxed boundary. Other than this boundary, $p$ and $\theta$ can be any values.

\subsection{Fixed points in the curve}
\label{sec:model:fixed}

The shape of the hyperbola is not easy to interpret from~\eqref{eq:hyperbola}, and hence it is not easy to say, by just looking at the parameters $p$ and $\theta$, whether the community is `fat' or `skinny'. To make the model parameters more interpretable, we can consider two points in the curve: the point at which it crosses the diagonal (i.e.\ when $i=j$), and the point at which the hyperbola exits the community (i.e.\ $j$ for which $i=n_C$ or vice versa). We call the former $\gamma$ and the latter $H$, and we can consider them as two values that define some $p$ and $\theta$ such that
\begin{align}
  (\gamma + p)(\gamma + p) &= \theta \label{eq:gamma_def} \\
  (H + p)(n_C - 1 + p) &= \theta\; . \label{eq:H_def}
\end{align}

\begin{figure}[tb]
	\centering
	\includegraphics[width=\singlefigwidth]{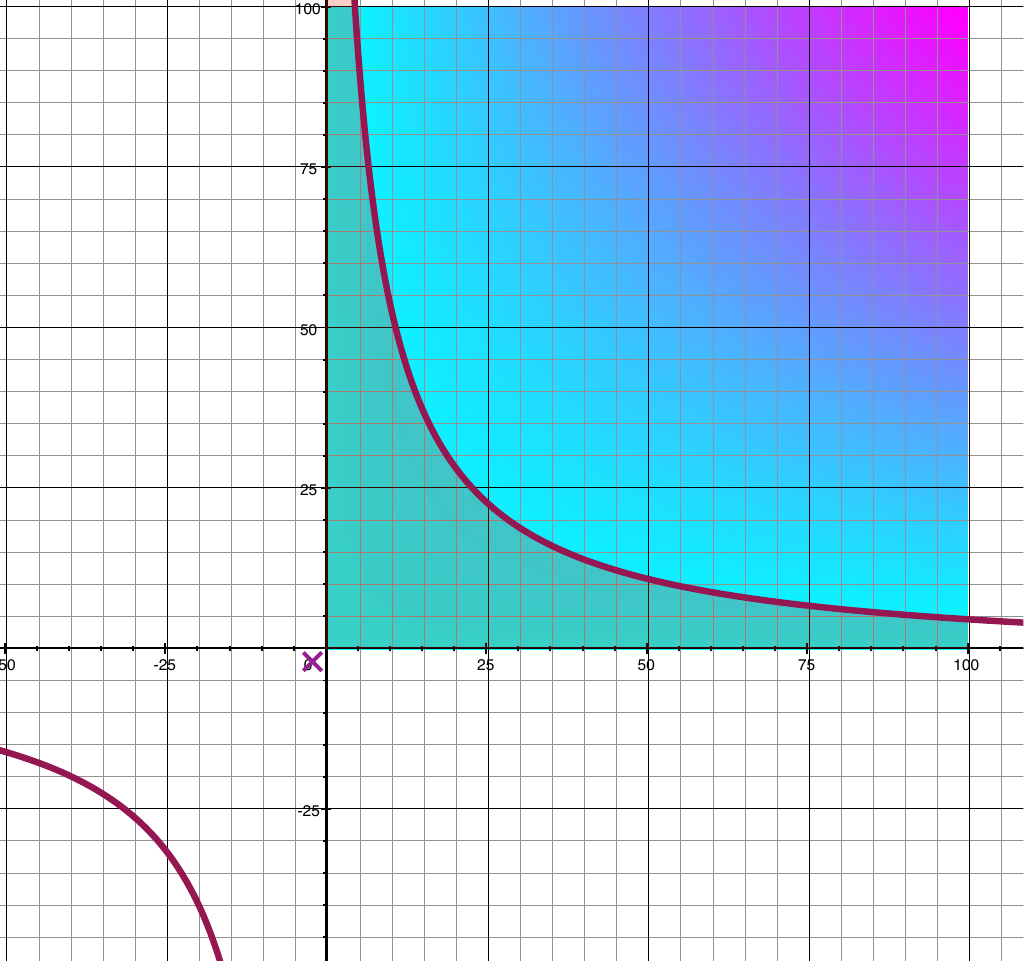}
	\caption{A hyperbola with $p = 2.06$ and $\theta = 673$ (dark red lines). The centre $(-p, -p)$ is marked with a cross. The colours in the background of the nonnegative quadrant indicate the values of $(i + p)(j + p)$ for $i,j \in [0, 99]$, with higher values moving from cyan to magenta. The area of the community is the solid-coloured area under the curve.}
	\label{fig:hyperbola}
\end{figure}

To interpret the parameters, it is helpful to divide every community into two parts: \emph{core} and \emph{tail}. The core consists of nodes that form a (quasi-) clique, while the tail consists of nodes that are mainly connected to the core. This is illustrated in Figure~\ref{fig:fixed_example} where the core is shaded in dark blue and the tail in light red. The parameter $\gamma$ is the size of the core (minus 1 to account for zero-based indexing) -- the larger $\gamma$, the larger clique the community has -- while the parameter $H$ tells how `fat' the tails are. A (quasi-) clique would have large $\gamma$ and $H$, while a star would have $\gamma = H = 0$.
We will call this model \mfixed{$\gamma, H$}. 

\begin{figure}[tb]
	\centering
	\includegraphics[width=\singlefigwidth]{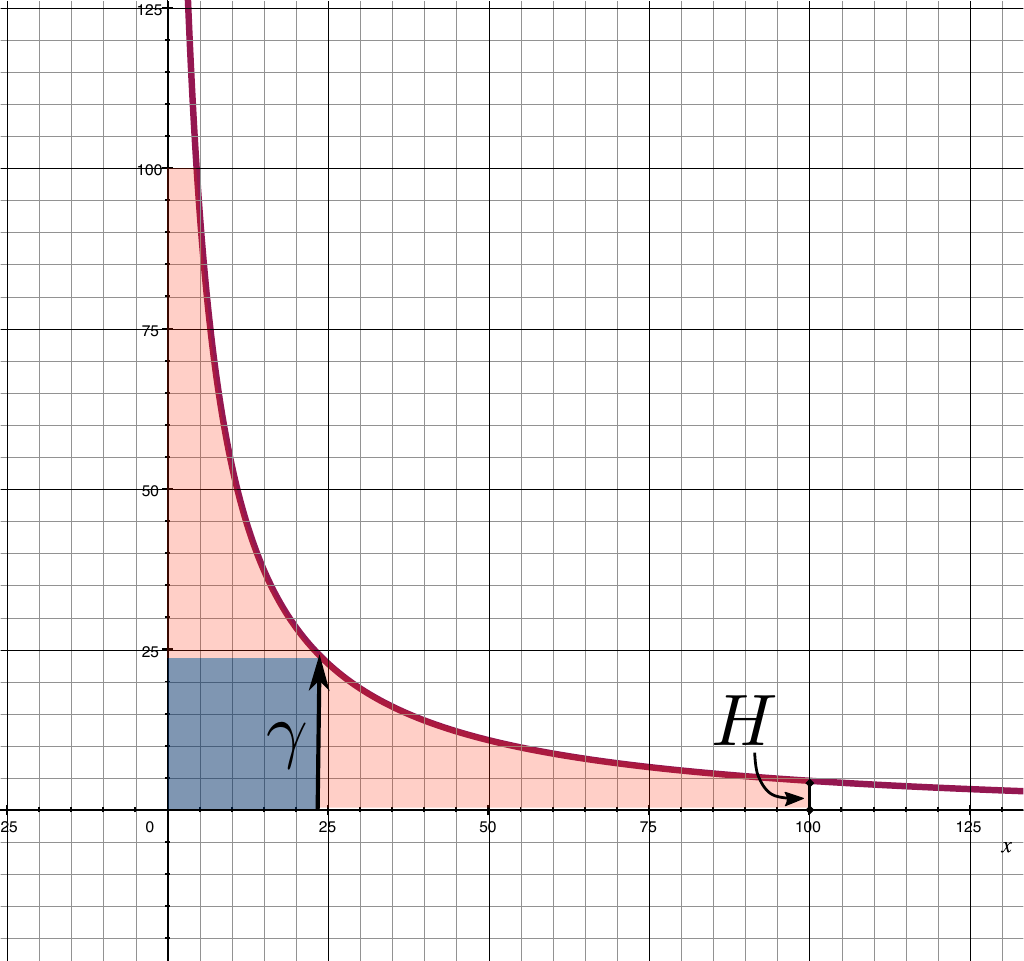}
	\caption{The parameter $\gamma$ explains the size of the core of the community (dark-shaded box), while $H$ explains the height of the tail at the end of the community.}
	\label{fig:fixed_example}
\end{figure}

\paragraph{Equivalence} Given equations \eqref{eq:gamma_def} and \eqref{eq:H_def}, it is hardly surprising that \mfixed{} is equivalent to \mhyperbolic{}. Formally, we can solve $p$ and $\theta$ given $\gamma$, $H$, and $n_C$ (the size of the community) as follows:
\begin{align}
  p &= \frac{\gamma^2 - (n_C-1)H}{(n_C-1) + H - 2\gamma}\label{eq:p_given_h_gamma} \\
  \theta &= \frac{(\gamma - H)^2(\gamma - n_C - 1)^2}{(n_C - 1 + H - 2\gamma)^2} \; .\label{eq:theta_given_h_gamma}
\end{align}
Similarly, we can easily work out the equations for $\gamma$ and $H$ given $p$ and $\theta$ (and $n_C$). Notice that these equations also give us a direct way to evaluate whether an edge $(i,j)$ is in a community \mfixed{$\gamma, H$}: we only need to solve $p$ and $\theta$ and evaluate if $(i+p)(j+p)\leq \theta$.

\paragraph{Constraints of parameters} Not every $\gamma$ and $H$ yield valid communities in \eqref{eq:p_given_h_gamma} and \eqref{eq:theta_given_h_gamma}, so we have to pay particular attention on the constraints of these parameters. Clearly both $\gamma$ and $H$ need to be nonnegative. 
As the hyperbola is monotonic, it must be that $H \leq \gamma$, and as the hyperbola is convex, it must be that $\gamma \leq (n_C - 1 + H)/2$.

Recall that in Section~\ref{sec:model:hyperbola} we restricted $p$ to be nonnegative so that it cannot happen that $(0+p)(0+p) > \theta$ if $(i+p)(j+p) \leq \theta$ for some $(i,j)\in\N\times\N$. The motivation for this was that we want element $(0,0)$ to be included in every non-empty community. But we can relax the constraint $p\geq 0$ to
\begin{equation}
  \label{eq:constrain_p}
  p^2\leq\theta \Leftrightarrow p^2 \leq (\gamma + p)^2 \Leftrightarrow p \geq -\frac{\gamma}{2} \; ,
\end{equation}
assuming $\gamma > 0$. Here, the first equivalence follows by substituting~\eqref{eq:gamma_def} to $\theta$. Notice that the $p$ and $\gamma$ in this inequality are bound together via~\eqref{eq:p_given_h_gamma}. This constraint also implies the above constraint that $H \leq \gamma$. 

\subsection{A mixture of a line and a hyperbola}
\label{sec:model:mixture}

Given the understanding of the previous two models, we now introduce a third equivalent model. Here, we consider a mixture of two restricted models: a simple hyperbola where an edge $(i,j)$ is in the community if $i \cdot j \leq \Sigma'$ for some $\Sigma'\in\R$, and a simple linear model $i+ j \leq \Sigma''$, with $\Sigma''\in\R$. Notice that unlike above, here the hyperbola's centre is fixed to the origin. Alone neither of these models is very expressive, but a mixture of these two is much more powerful: an edge $(i, j)$ is in the community if
\begin{equation*}
  \label{eq:mixtureA}
  (1 - x)(i\cdot j) + x(i + j) \leq \Sigma
\end{equation*}
for some $x\in[0, 1]$ and $\Sigma\in\R$.
Indeed, to allow even more flexibility, one can also consider a second mixture, which combines the hyperbola with a negative linear model
\begin{equation*}
	\label{eq:mixtureB}
	(1 - x)(i\cdot j) + x(-i - j) \leq \Sigma
\end{equation*}
for some $x\in[0, 1]$. Combining  both of the above equations into a single model leads to our final definition: an edge $(i, j)$ is in the community if for some $x\in[-1, 1]$ and $\Sigma\in\R$
\begin{equation}
	\label{eq:mixture}
	(1 - |x|)(i\cdot j) + x(i + j) \leq \Sigma \; .
\end{equation}

The meaning of the parameters here is slightly different to the model \mhyperbolic{}: the parameter $\Sigma$ behaves similarly to $\theta$ in \eqref{eq:hyperbola}, moving the line (or the hyperbola) further from the origin, while the mixture parameter $x$ dictates how much the model looks like a line and how much like a hyperbola centered to the origin. We call this model \mmixture{$x, \Sigma$}.

\paragraph{Equivalence}
We will now turn our attention to the equivalence between \mhyperbolic{} and \mmixture{}. 
\begin{proposition}
	\label{prop:hyper_mix_equivalent}
	For any pair $(p, \theta)$ of valid parameters of the \mhyperbolic{} model, there is a pair $(x, \Sigma)$ of valid parameters of the \mmixture{} model that yield exactly the same community, and vice versa.  
\end{proposition}
\begin{proof}
	We have\\[-4mm]
	\begin{equation}
	\label{eq:4}
	\begin{split}
	(i+p)(j+p) &\leq \theta \\
	\Leftrightarrow \frac{(i+j)p}{1+|p|} + \frac{ij}{1+|p|} &\leq \frac{\theta - p^2}{1 + |p|} \\
	\Leftrightarrow (i+j)\frac{p}{1+|p|} + ij\left(1 - \Big|\frac{p}{1+|p|}\Big|\right) &\leq \frac{\theta - p^2}{1+|p|} \; ,
	\end{split}
	\end{equation}
	where the first equivalence is by expanding the left-hand side and re-arranging the terms and the second is by noting that
	$$\frac{1}{1+|p|}=\frac{1+|p|-|p|}{1+|p|}= 1 - \frac{|p|}{1+|p|} = 1 - \Big|\frac{p}{1+|p|}\Big| $$
	If we now write $x = p/(1+|p|)$ and $\Sigma = (\theta - p^2)/(1+|p|)$, we get
	\begin{equation}
	\label{eq:5}
	(i+j)x + ij(1-|x|) \leq \Sigma\; ,
	\end{equation}
	concluding the proof. 
\end{proof}

\subsection{Generalization of Quasi-Cliques and HyCom}\label{sec:hycom}

An important consideration in our model(s) is that we want to be able to generalize existing models.

\paragraph{Quasi-cliques}
 Clearly, quasi-cliques are a special case of our model.
 Using  the \mfixed{} model, we can simply set $H$ and $\gamma$ to $n_C$.
 
\paragraph{HyCom}
In \cite{araujo2014beyond}, a community has been defined as follows: Given $\alpha<0$ and $\tau \in \R$, all edges $(u,v)$ that fulfill
\begin{equation}
	\label{eq:mixtureHycom}
	u^\alpha\cdot v^\alpha  \geq \tau
\end{equation}
are part of the community. Note that in \cite{araujo2014beyond}, one-based indexing is used, i.e.\ the indices of nodes $u,v$ start with 1. Thus, using our zero-based indexing w.r.t.\ $i,j$, \eqref{eq:mixtureHycom} is equivalent to
\begin{equation}
\label{eq:mixtureHycom2}
(i+1)^\alpha\cdot (j+1)^\alpha  \geq \tau
\,\Leftrightarrow\,
0.5\cdot (i\cdot j) + 0.5\cdot (i+j) \leq \tau' \; .
\end{equation}
Here, we set  $\tau'=0.5\cdot (\tau^{1/\alpha}-1)$ and exploited that $\alpha < 0$.

Using our \mmixture{} model, it is now obvious that \cite{araujo2014beyond} is a special case: it corresponds to \mmixture{$0.5, \tau'$}. Indeed, while at first sight \eqref{eq:mixtureHycom} seems to have two degrees of freedom, it only has one.  The parameter $\tau'$ (in our notation: $\Sigma$) can be adapted, the parameter $x$ is fixed to $0.5$. 


This restriction in the HyCom model limits the possible shapes of the communities significantly since communities with $x\neq0.5$ can not be represented.  
Indeed, as we will show in the experiments, many real world datasets contain communities with $x$ not close to $0.5$ (see Figure~\ref{fig:distx}); thus, they are much better represented by our model. Furthermore, the HyCom model does not give us intuitive interpretation of the shape of the community as the exponent $\alpha$ can be switched to any other exponent by adjusting $\tau$ accordingly.


\section{Model for the Full Graph}
\label{sec:full-graph}
While the previous section focused on modelling individual communities, we now introduce a principle for describing the full graph containing multiple communities. When multiple communities are present, we might observe overlapping groups. When communities are modelled as quasi-cliques, there is only one type of overlap we need to consider: if the nodes of two communities overlap, so do their (implicit) edges. With our community models, however, we have to distinguish three types of overlapping behaviour: two communities $C$ and $D$ are \emph{node-disjoint} if they do not share any nodes ($V_C \cap V_D = \emptyset$); \emph{area-disjoint} (but \emph{node-overlapping}) if they do share nodes but no (implicit) edges ($V_C \cap V_D \neq \emptyset$ but $A_C \cap A_D = \emptyset$); and \emph{area-overlapping} (or \emph{overlapping} for short) if also their (implicit) edges overlap ($A_C \cap A_D \neq \emptyset$). 


Area-overlapping communities present a particular challenge to the modelling as we have to assign a likelihood to every (implicit) edge. In this work we assign each edge to at most one community and the likelihood of the edge is calculated using only that community. 

For defining the quality of a set of communities, we refer to a probabilistic approach, i.e.\ we aim to find the set of models $\col{C}$ leading to the highest \emph{likelihood} of the input graph $G$. For this purpose, notice that in real graphs the communities rarely have full density (i.e.\ $\abs{E_C}\neq\abs{A_C}$), and that the density varies between the communities.

We use the model from Section~\ref{sec:model:general} as the basis for modelling a community. That is, the density of a community $C$, $d_C$, is defined as in~\eqref{eq:d_C}, with every community having its own density. To define the outside density $d_O$, we have to consider not only the `outside area' of a single community, but the whole area of the graph that does not belong to any community. 
%
%
%
If we let $A_{\col{C}} = \cup_{C\in\col{C}} A_C$ be the area that belongs to the communities and let $\complement{A_{\col{C}}} = (V\times V) \setminus A_{\col{C}}$ be its complement, then 
\[
d_O = \abs{E \setminus A_{\col{C}}}/\abs{\complement{A_{\col{C}}}}\; .
\] 


We can now model the full graph similarly to how we modelled a single community to obtain the overall likelihood $L(G\mid \col{C})$ of a graph $G$ given the set of communities $\col{C}$.

\begin{definition}
  Given a graph $G$ and a collection of its communities $\col{C}$, where every edge belongs to at most one community, the \emph{log-likelihood} $\log L(G\mid \col{C})$ is defined as 
  \begin{equation}
    \label{eq:loglikelihood}
    \begin{split}
      \log L(G\mid \col{C}) &= 
      \sum_{C\in\col C} \bigl(\abs{E_C}\log(d_C) + \abs*{\complement{E_C}}\log(1-d_C)\bigr) \\
      &\quad + \abs{E \setminus A_{\col{C}}}\log(d_O) \\
      &\quad + (\abs{\complement{A_{\col{C}}}} - \abs{E \setminus A_{\col{C}}})\log(1 - d_O) \; .
    \end{split}
  \end{equation}
\end{definition}




\section{Computational Complexity}
\label{sec:complexity}

Before we present our algorithms, let us briefly study the computational complexity of the problems related to the modelling. Instead of dealing directly with the likelihood, in this section our target is to minimize the number of non-edges inside the communities while simultaneously maximizing the number of non-edges outside (i.e.\ to maximize the community density while minimizing the outside-area density). This is a natural surrogate for the likelihood that allows us to avoid some issues in the analysis caused by the likelihood function (e.g.\ that the likelihood model is oblivious to the `inside' and `outside': very sparse communities with dense outside area are also good models in likelihood's sense). 

We will first consider problems involving only a single community, showing that finding the node sets for communities is hard in our model:

\begin{proposition}
  \label{prop:1com:fixed_param_hard} Given a graph $G=(V, E)$ and a pair of parameters $(p, \theta)$, it is \NP-hard to find 
  \begin{itemize}
  \item the largest set of nodes $V_C\subset V$ and a permutation $\pi_C: V_C\to \{0,\ldots,\abs{V_C}-1\}$ such that the area $A_C$ defined by $V_C$, $\pi_C$, and \mhyperbolic{$p, \theta$} is \emph{exact}, that is $A_C = E_C$;
  \item the set of nodes $V_C\subset V$ and a permutation $\pi_C : V_C\to \{0,\ldots,\abs{V_C}-1\}$ such that $d_C\geq c$ for some given constant $c\in(0,1)$ and $d_O$ is minimized.
  \end{itemize}
\end{proposition}

\begin{proof}
  These results follow from the fact that the clique is a special case of our model. Hence, if $(p, \theta)$ are set so that they encode a clique, the first case is equivalent to the well-known \NP-hard problem of finding the largest clique~\cite[Problem GT19]{garey79computers}, while the second is equivalent to the problem of finding the maximum $c$-quasi-clique, which is also \NP-hard~\cite{pattillo13maximum}.
\end{proof}





Let us now turn our attention to the case where we are already given a collection $\col{C}$ of communities (with fitted models), and we want to find a subcollection $\col{S}\subseteq\col{C}$ that minimizes the number of edges in the outside area plus the number of non-edges inside the communities. That is, we want to minimize 
\begin{equation}
  \label{eq:min_err}
  \abs{E \cap \complement{A_{\col{S}}}} + \abs{A_{\col{S}} \setminus E}  \; .
\end{equation}

For these results, we use the general framework of Miettinen~\cite{miettinen15generalized}. First note, that our communities are (symmetric) generalized rank-1 matrices in the sense of Miettinen: the functions $f$ of our models define the outer product in Definition 1 of \cite{miettinen15generalized}, while the adjacency matrix is the data matrix. For this problem, we only care about the union of the areas of the communities, and consequently, we take the element-wise disjunction of the matrices representing their areas. Propositions 6 and 10 of \cite{miettinen15generalized} directly provide the following results:
\begin{proposition}
  \label{prop:sub_comm_hard}
  Given a graph $G=(V, E)$ and a collection $\col{C}$ of communities of $G$, 
  \begin{itemize}
  \item it is \NP-hard to find the subcollection $\col{S}\subseteq\col{C}$ that minimizes \eqref{eq:min_err};
  \item it is \NP-hard to approximate the error \eqref{eq:min_err} to within a factor of $\Omega\left(2^{\log^{1-\varepsilon}\abs{V}}\right)$ and quasi-\NP-hard to approximate it within $\Omega\left(2^{(4\log\abs{\col{C}})^{1-\varepsilon}}\right)$ for any $\varepsilon > 0$;
  \item the error \eqref{eq:min_err} can be approximated to within a factor of $2\sqrt{(\abs{\col{C}} + \abs{V})\log\abs{V}}$ in polynomial time.
  \end{itemize}
\end{proposition}

The situation of Proposition~\ref{prop:sub_comm_hard} can easily arise as the consequence of the following simple idea of finding the communities: first, find many subset of nodes (e.g.\ by sampling or by enumerating all dense subgraphs), then fit the community models to them, and then select the final subset of communities from the potentially highly redundant set of communities. As Proposition~\ref{prop:sub_comm_hard} shows, the last step of this approach is computationally hard and hence we do not use it.


\section{Algorithms}
\label{sec:algorithms}

Next we present an algorithm for fitting our model to a graph. We assume the input is  the graph and an initial collection of sets of nodes that represent initial communities. These initial node sets can be found using any existing community-detection algorithm, e.g.\ HyCom \cite{araujo2014beyond}.
We will first present the algorithm to fit the model to a single community, and then explain how to use that to fit the model for the full graph.

\subsection{Modelling a single community}

To model a single community, we use the \mfixed{} model with \emph{integer} parameters $\gamma$ and $H$. Given the intuition from Figure~\ref{fig:fixed_example}, this restriction is natural. We will also not lose too much, as the following proposition demonstrates:
\begin{proposition}
  \label{prop:area_with_varying_gamma} Let $C = \mfixed{\gamma, H}$ be a community of $n_C$ nodes defined by $\gamma\in\R$ and $H\in\N$, and let $A_C$ be its area. Then there exists integer $\gamma'$ such that if $D = \mfixed{\gamma', H}$ is the community defined by $\gamma'$ and $H$, and $A_D$ is the respective area, then $\abs{A_C - A_D} \in \Theta(\gamma \ln(n_C))$.
\end{proposition}

In other words, the difference in area between integer and non-integer $\gamma$ grows only linearly with $\gamma$ and logarithmically with the number of nodes in the community. The proof of Proposition~\ref{prop:area_with_varying_gamma} is postponed to the appendix.

Constraining ourselves to integer parameters would not alone solve much, as there still are $O(n_C^2)$ parameter configurations to study. Many of these configurations, however, can be pruned out as they would lead to infeasible communities and the pruned search space is small enough for exhaustive search.



To gain intuition on how much of the search space the constraints remove, let us consider Figure~\ref{fig:feasible_parameters}. The area of the plot is the area of all possible combinations of $\gamma$ and $H$ for some community size $n_C$. The yellow area can be ignored as in that area $H > \gamma$. But also both of the green areas can be ignored, as in those areas, the constraint $p \geq -\gamma/2$ is violated (see \eqref{eq:constrain_p}). This pruning significantly reduces the different parameter configurations we need to test in the exhaustive search. 

\begin{figure}[tb]
	\centering
	\includegraphics[width=\singlefigwidth]{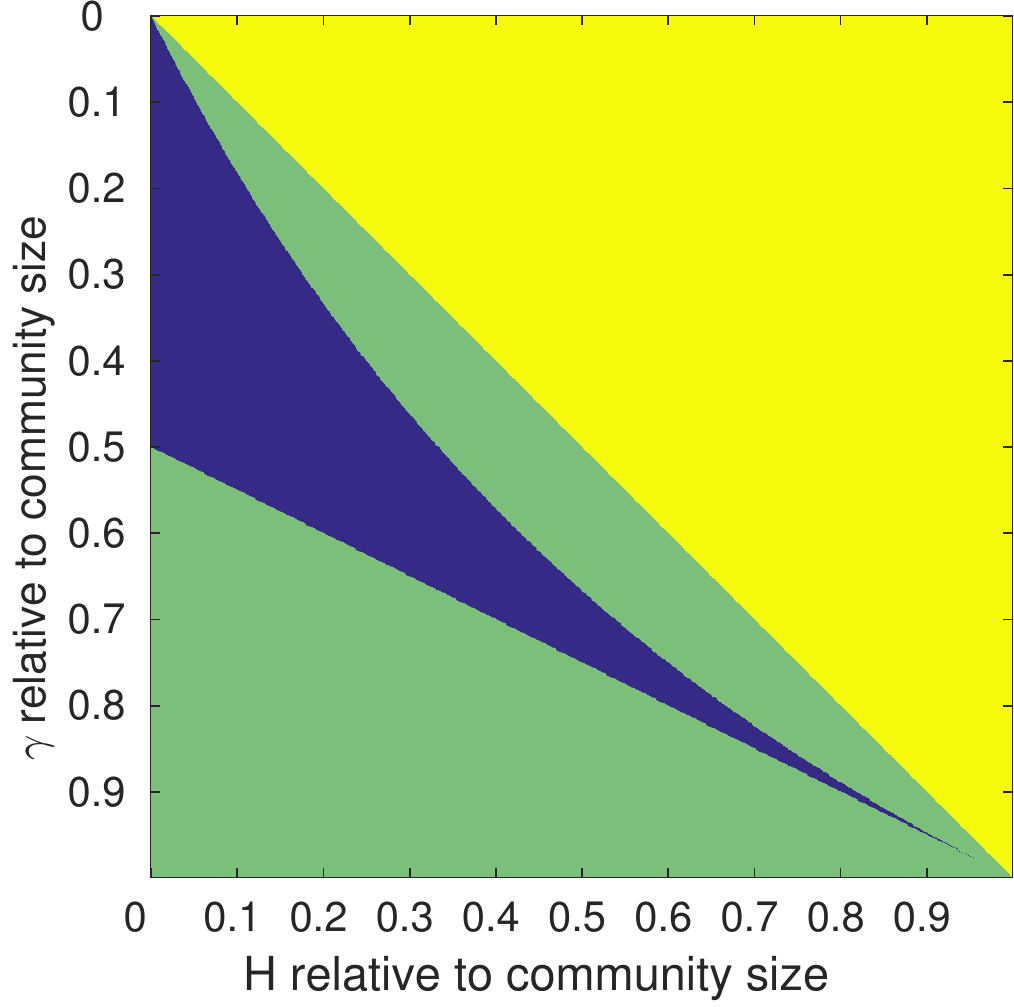}
	\caption{The blue area shows feasible values for parameters $\gamma$ and $H$ relative to a given community size. The green and yellow area are infeasible. The yellow part violates the trivial requirement of $H\leq\gamma$ and the green areas violate the condition $p \geq -\frac{\gamma}{2}$, where $p$ is given by Equation \eqref{eq:p_given_h_gamma}. }
	\label{fig:feasible_parameters}
\end{figure}

\paragraph{Likelihood-Computation} Given a pair of parameters $(\gamma, H)$, we need to evaluate its fit by computing the log-likelihood of the resulting model. According to the model proposed in Section~\ref{sec:model:general}, this requires to determine the area inside and outside the current community as well as the corresponding number of edges (and non-edges) in these areas.
Obviously, testing each position in the community is not a practical solution since it would lead to a running time quadratic in the community's size; instead  we derive a solution which is linear in the number of edges.

The computation of the area can be done in time linear in the number of nodes by referring to the functional form of the hyperbola, i.e.\ evaluating 
$
i = \frac{\theta}{j+p} - p
$
for each column $j$. Here we can compute $p$ and $\theta$ from $\gamma$ and $H$ using \eqref{eq:p_given_h_gamma} and \eqref{eq:theta_given_h_gamma}. Alternatively, we can approximate the area in constant time by taking the integral of this function from $0$ to $n_C -1$. Counting how many edges are inside the community requires a pass over the edges. Thus, this step dominates the time complexity.

It is easy to optimize this procedure further: First, we can compute the area faster by noticing that at the bottom we have a rectangle of size \by{n_C}{H}. Second, when we test a succession of parameter values, we can re-use part of the information about the edges: by increasing the values of $H$ or $\gamma$ only edges previously outside the community need to be evaluated. All remaining edges will still be located within the community.

\subsection{Modelling a full graph}
\label{sec:alg:full}
To obtain a model for a graph consisting of multiple potentially overlapping communities, we aim to optimize the log-likelihood $L(G\mid\col{C})$ for the full graph (i.e.~\eqref{eq:loglikelihood}). Our main problem is to determine how to deal with overlapping communities; indeed, if there are no node-overlapping communities, we can simply optimize every community separately using the above approach. If the communities do overlap, however, we do need to decide the order of the communities so that we can assign every edge to at most one community. 

As computing the log-likelihood for the full graph for each possible order of the communities is infeasible, we optimize the log-likelihoods of each community individually following an alternating optimization strategy. When optimizing one community, we keep track of the area that was already covered by other communities, and ignore that area in the computations of the subsequent communities.
In concordance with the log-likelihood we want to optimize for, we consider the global density for the whole outside community area (see Section~\ref{sec:full-graph}) in the log-likelihood computation of the model, instead of just the local outside density, as in Section~\ref{sec:model:general}.

The algorithm, shown as Algorithm~\ref{alg:full_graph}, comprises an initialization phase and an update phase. In the initialization phase of the algorithm (lines \ref{alg:full:init_start}--\ref{alg:full:init_stop}), we compute an individual model for each community, leaving out those edges that have already been covered in a previous step by another community. Note that during this step, each community uses its individual outside density. Next, we order the obtained models by their log-likelihood starting with the best and we compute the global outside density $d_O$ for the further updates.

Now that we have established an order of the communities, the alternating optimization starts (line~\ref{alg:full:iter_start}): Each time a community $C_i$ is selected and a new model is fit to it -- now not only excluding edges already covered in previous communities but also using the global outside density to determine the true log-likelihood for each community. After fitting the new model, we update the global outside density (line~\ref{alg:full:update_d_O}) if the new model is different to the old one. 

All communities that have node overlap with $C_i$ are marked: due to the update of $C_i$ also the parameters of the overlapping communities might change. Thus, we mark the communities that overlap with $C_i$ for a re-update (line~\ref{alg:full:flag}).
We iterate over this process of updating the community models until there are no more communities to be updated.

The output of this algorithm is a list of models for all communities, ordered by their log-likelihoods.

\begin{algorithm}[tb]\small
  \caption{Algorithm to fit the \mfixed{} model to a graph.}\label{alg:full_graph}
  \begin{algorithmic}[1]
    \Input Undirected graph $G=(V, E)$, a collection of sets of nodes $\col{V} = \{V_i\subset V\}$ describing the initial communities
    \Output Ordered set of communities $\col{C}$
    \For{every set $V_i\in\col{V}$} \label{alg:full:init_start}
      \State $C \gets $ best model describing $G|_{V_i}$
      \State $\col{C}\gets \col{C} \cup \{C\}$
    \EndFor
    \State Order $\col{C}$ based on the likelihoods of the communities 
    \State Compute the global outside density $d_O$ \label{alg:full:init_stop}
    \State $\col{F}\gets \col{C}$
    \Repeat \label{alg:full:iter_start}
      \State $\col{T}\gets\col{F}$; $\col{F}\gets\emptyset$; $M\gets \emptyset$
      \ForAll{$C \in \col{T}$} \Comment{In decreasing likelihood}
        \State Update the model of $C$ ignoring areas in $M$
        \State $M \gets M \cup A_C$ \label{alg:full:area}
        \If{the likelihood of $C$ improved}
          \State Update $d_O$ \label{alg:full:update_d_O}
          \State $\col{F} \gets \col{F}\cup \{D\in\col{C} : V_D \cap V_C \neq \emptyset\}$ \label{alg:full:flag}
          \State Update the position of $C$ in $\col{C}$
        \EndIf
      \EndFor
    \Until{$\col{F} = \emptyset$}
    \State \textbf{return} $\col{C}$
  \end{algorithmic}
\end{algorithm}


\section{Experiments}
\label{sec:experiments}
We divide our experiments into two groups. In the first group (Section~\ref{sec:experiments:obtainedModels} and~\ref{sec:comparisonBlockHyCom}), we use graphs with human-annotated communities and our goal is to study the shape of these communities, and the differences between the block/HyCom models and our model. In the second group (Section~\ref{sec:experiments:findingCommunities}) of experiments, we use existing community-detection algorithms to find the communities, and then apply our model to the found communities. The source code for our programs and the scripts to run the experiments are available online.\!\footnote{\url{http://people.mpi-inf.mpg.de/~pmiettin/hybobo/}}

\subsection{Obtained models}\label{sec:experiments:obtainedModels}
We start by studying the results of fitting our model to graphs where ground-truth communities are given. We used the Stanford Large Network Dataset collection~\cite{leskovec14snap}, which offers datasets with known communities. Table~\ref{tab:datasets} provides a summary of the employed datasets. As neither very small nor very large communities are particularly interesting, we restrict our analysis to communities with between 100 and 1000 nodes (inclusive). We use a sample of 500 communities from each of these data sets (for \dblp we also use a sample of 100 communities, see Section~\ref{sec:comparisonBlockHyCom}). Notice that the ground-truth information for the communities is only provided for the nodes and not for the edges. 

\begin{table}[t]
	\caption{Sizes of the datasets used for evaluation, their total number of communities as well as the number of communities of size 100 to 1000, and the time it took to determine a hyperbolic model for a sample of size 500.}
	\label{tab:datasets}
	\centering
	\begin{tabular}{@{}l@{\hspace{0.4em}}r@{\hspace{0.6em}}r@{\hspace{0.6em}}r@{\hspace{0.0em}}r@{\hspace{0.2em}}r@{}}
		\toprule
		& \multicolumn{1}{c}{nodes} &  \multicolumn{1}{c}{edges} & \multicolumn{2}{c}{communities} & \multicolumn{1}{c}{time} \\ 																												
	\cmidrule(rl){4-5}
			&  &   & \multicolumn{1}{c}{all} & \multicolumn{1}{c}{100--1000}&\multicolumn{1}{c}{(h)}\\ 
			\midrule
		\amazon & 334,863 &	925,872 &	75,149& 1,380& 0.6 \\ 
		\dblp & 317,080 &	1,049,866 &	13,477 & 805& 27.0\\ 
		\friendster & 65,608,366 &	1,806,067,135 &	957,154 & 19,763& 12.3\\ 
		\lj &  3,997,962 &	34,681,189 &	287,512 & 8,769& 11.3\\ 
		\orkut &  3,072,441 &	117,185,083 &	6,288,363 & 80,251& 3.1 \\ 
		\youtube &  1,134,890 &	2,987,624 &	8,385 & 129& 0.8\\ 
		\bottomrule 
	\end{tabular}
\end{table}

\begin{figure*}[tb]
	\centering
	\subfloat[Distribution of $\gamma$.]{
		\centering
		\includegraphics[width=0.3\textwidth]{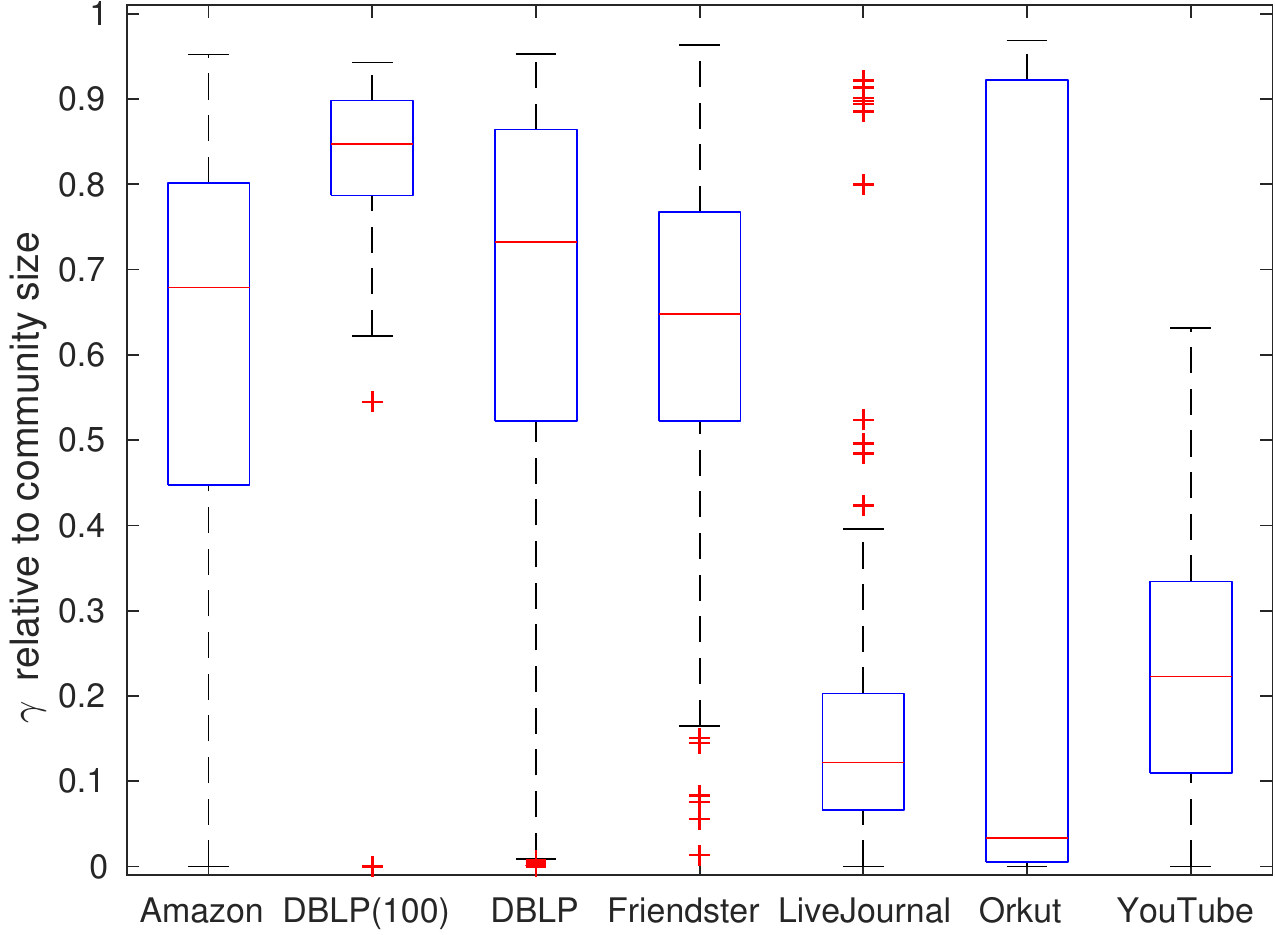}
		\label{fig:distgamma}
	}
	\subfloat[Distribution of H.]{	
		\centering
		\includegraphics[width=0.3\textwidth]{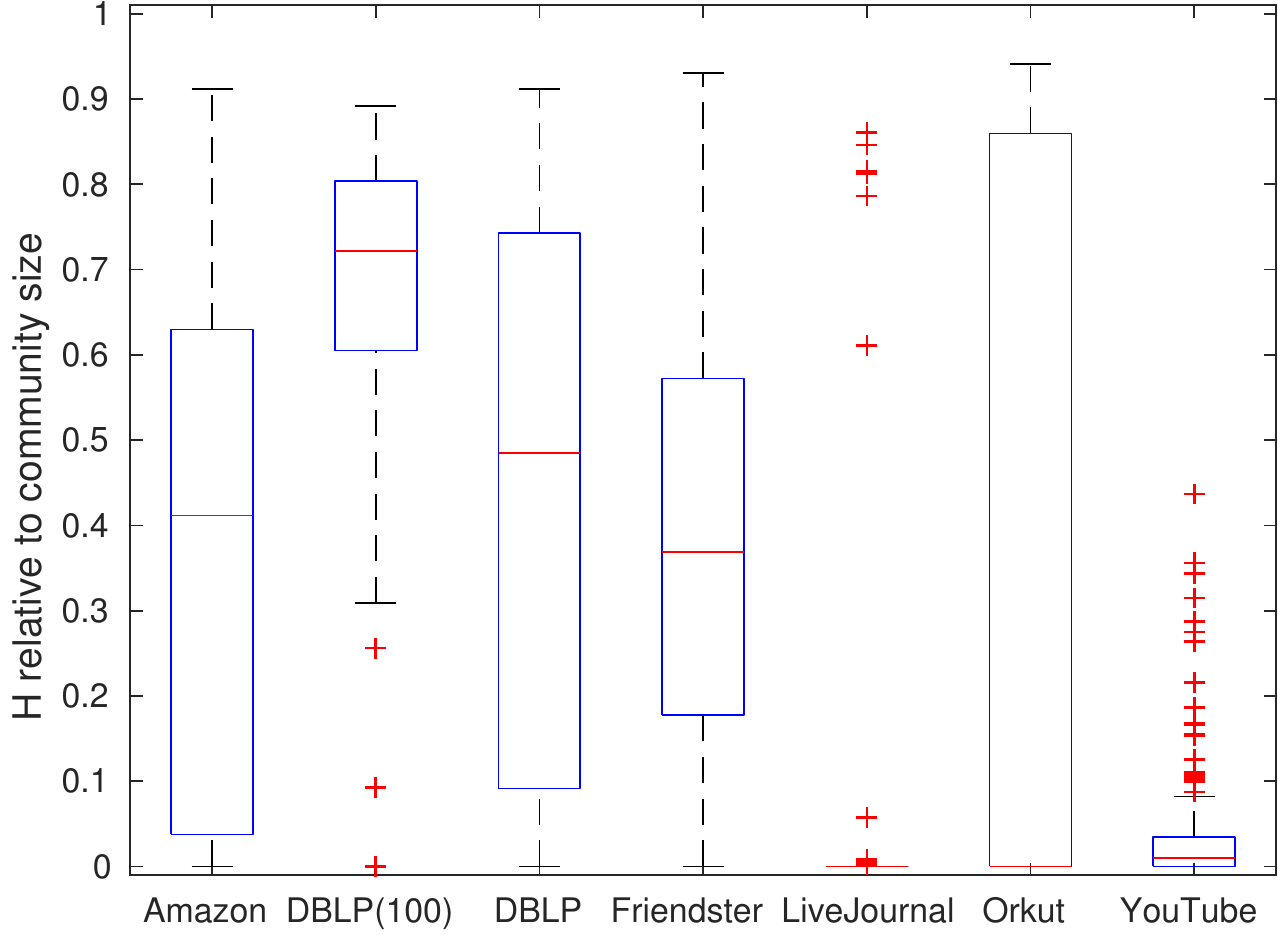}
		\label{fig:distH}
	}
	\subfloat[Distribution of $x$.]{
		\centering
		\includegraphics[width=0.3\textwidth]{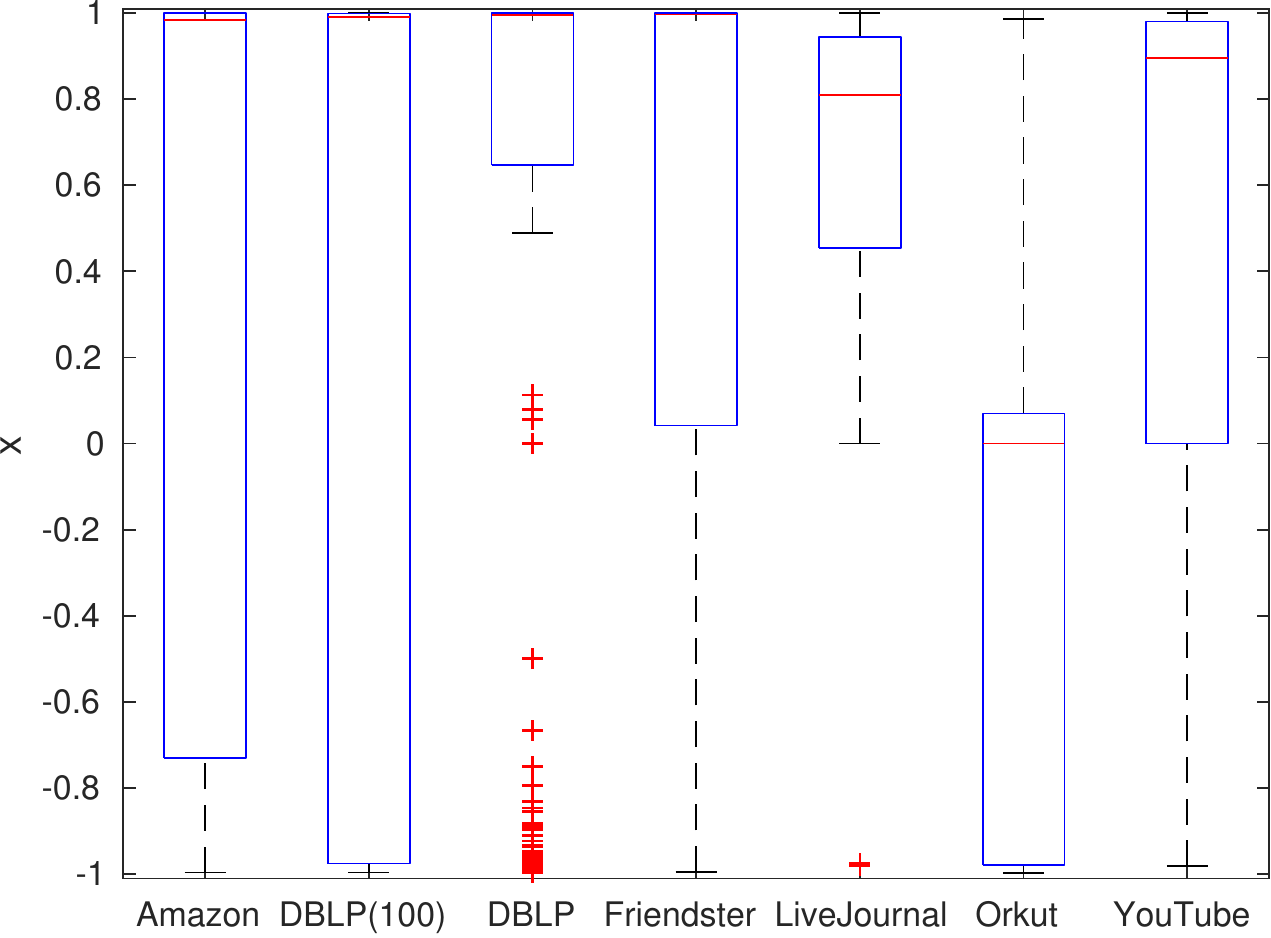}
		\label{fig:distx}
	}
	\caption{Distribution of selected parameters after fitting a hyperbolic model to  sampled communities of size 100 to 1000 from each dataset. The sample size was 500 for all datasets. For \dblp, the result is shown also for a sample of 100 communities.}
	\label{fig:distOfParams}
\end{figure*}


Figure~\ref{fig:distOfParams} summarizes the distributions of the parameters $\gamma$, $H$, and $x$ for different data sets, respectively. The latter parameter we infer by converting the results into the \mmixture{} parametrization. For all box plots, the median is indicated as the central red mark and the edges of each box are the 25 and 75 percentiles. The whiskers extend to the most extreme data points not considered outliers. 

The box plots reveal a characteristic shape for each data set: \amazon, \dblp, and \friendster show rather thick communities with $\gamma$ on median being more than $50\%$ of the community size. \lj, \orkut, and \youtube communities mostly have thin cores. While \orkut also has communities with big cores, \lj is at the other extreme and its communities mostly exhibit a star-like shape. These observations are in line with the intuitive idea about some of the data sets. \dblp, for instance, is the network of co-authorships and therefore many communities are almost quasi-cliques. 

Furthermore, we observe that the obtained models for all datasets differ significantly from the HyCom model that would require $x=0.5$ (see Section~\ref{sec:hycom}). As Figure~\ref{fig:distx} shows, none of the medians is near $x=0.5$.
To gain intuition on how these communities look like, we give a few examples in Figures~\ref{fig:exampleCommYoutube_both} and~\ref{fig:amazon_example}, for the \youtube and \amazon  data, and further examples in the appendix.

Our algorithm is implemented in Matlab and C. It took between half an hour and a full day to compute the models on a machine with four Intel Xeon E7-4860 10-core CPUs running at 2.27\,GHz and 256\,GB of main memory. The exact running times are given in Table~\ref{tab:datasets}. Notice that the time depends not only on the size of the graphs, but to a great extend on the amount of overlap between the communities. Overlapping communities have to be computed in a sequential manner as the model of one community has impact on the next. Additionally, an update to one model causes the re-computation of all communities that overlap with it. Hence, the computation for \dblp takes more than twice as much time than that for \friendster although the \friendster graph is orders of magnitude larger.


\begin{figure}[tb]
	\centering
	\subfloat[Before.]{
		\centering
		\includegraphics[width=0.472 \linewidth]{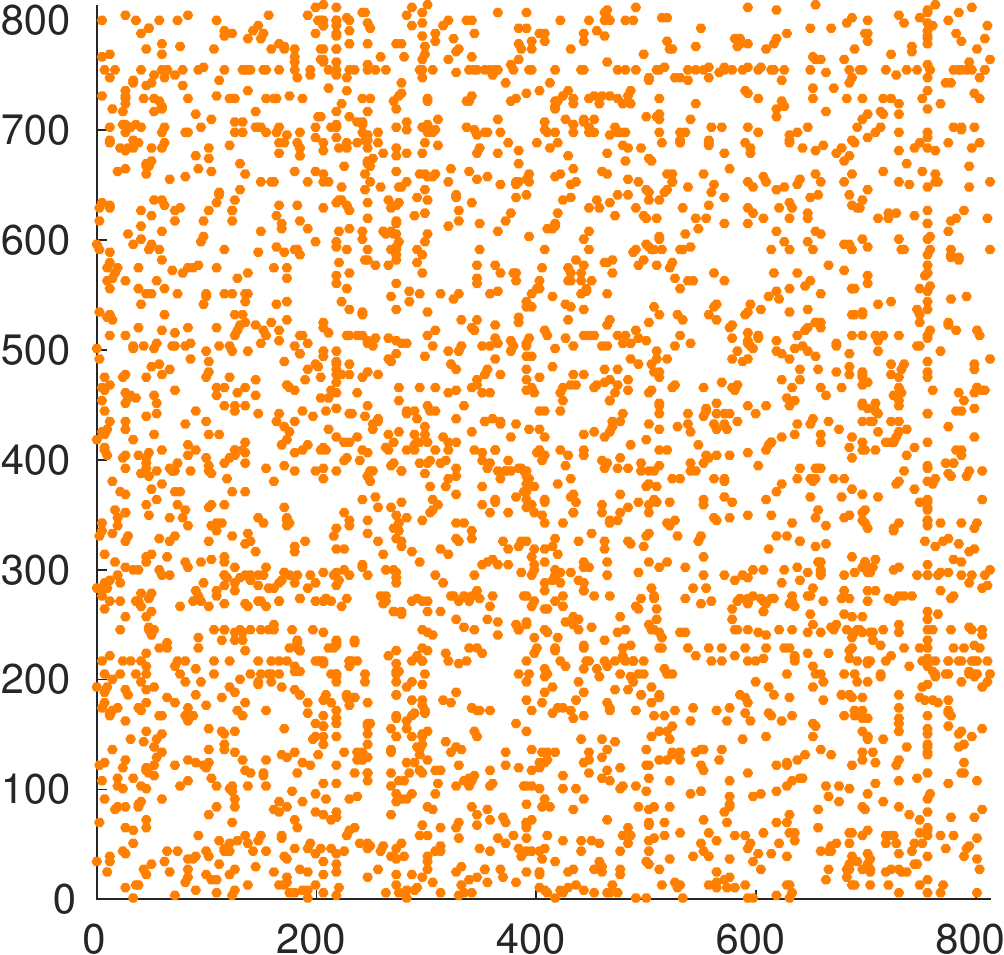}
		\label{fig:amazon_before}
	}%
	\subfloat[After fitting the model.]{
		\centering
		\includegraphics[width=0.472 \linewidth]{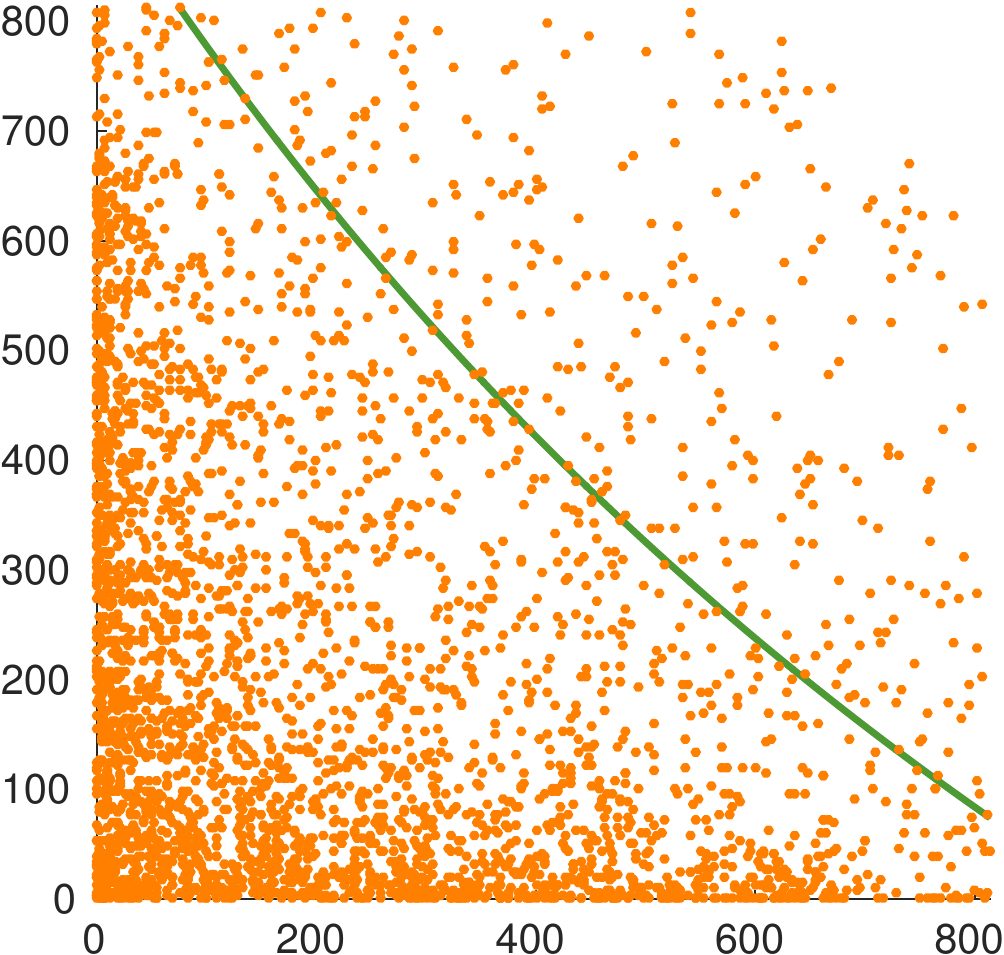}
		\label{fig:amazon_after}
	}
	\caption{Example of a community obtained from \amazon data.}
	\label{fig:amazon_example}
\end{figure}

\subsection{Comparison to block models and HyCom models}
\label{sec:comparisonBlockHyCom}

In Section~\ref{sec:hycom}, we describe two special cases of our model: In one case, every community is assumed to be a quasi-clique yielding exactly one possible parameter configuration per community, i.e.\ $H=\gamma=n_C$. In the other case, the HyCom model, only the threshold parameter $\Sigma$ may vary and the other is fixed to $x=0.5$. In this section, we assess the benefit of the additional flexibility of our model by comparing to these restricted versions. To do so, we compare the log-likelihood, computed as in~\eqref{eq:loglikelihood}, of the hyperbolic models obtained as in Section~\ref{sec:experiments:obtainedModels} to the log-likelihoods of the respective HyCom models and block models. The HyCom models of every data set were obtained by running Algorithm~\ref{alg:full_graph} but only admitting those $(H,\gamma)$ combinations that yield $x=0.5$. For the block models, no parameter search is necessary as there is only one admitted configuration per community.

To compare the log-likelihoods, we use the likelihood ratio test~\cite[Ch.~10.6]{wasserman10statistics}. 
In case of the block model, we test the null hypothesis $H_0$ that all parameters, i.e.\ all $H$ and $\gamma$ of all communities, are fixed to create the blocks versus the alternative hypothesis $H_1$ that the parameters are not fixed; for HyCom, we assume one free parameter. The likelihood ratio test statistics are given by
$\lambda=2 \log \bigl(L(\text{our model})/L(\text{block model})\bigr)$ and
$\lambda=2 \log \bigl(L(\text{our model})/L(\text{HyCom model})\bigr)$ for block and HyCom models, respectively.  The results are shown in Table~\ref{tab:LLratioTest}.

\paragraph{Block model} For the block model, the derived $p$-values are with one exception always essentially zero and confirm that the hyperbolic model is statistically significantly better than the block model.
The exception is the 500 sample of the \dblp data. For this data set, the block model gives a better likelihood than ours. While cliques are a special case of our model, and hence we can always model each community as a clique, the iterative method to update the model parameters (Section~\ref{sec:alg:full}) is based on a greedy heuristic. In case of the \dblp data, the greedy heuristic has reached a local optimum that is less good than what could be obtained with pure block models. On one hand, this is partially because \dblp contains block-like communities, and on the other hand, the large overlaps between these communities might have lead the optimization astray. To test the latter hypothesis, we also sampled just 100 communities from \dblp: this reduced the amount of overlap between the communities (from $33\%$ to $18\%$) and also significantly improved our algorithm's results.


\paragraph{HyCom} For the comparison to the HyCom model, we obtain a similar result: With one exception, our hyperbolic model describes the data statistically significantly better than the HyCom model. The better solution for the 100 sample of \dblp found with the parameter space restricted to HyCom models is also a valid solution within our more general modelling framework. The greedy algorithm we propose, however, gives no guarantee to converge to the globally best solution and this result indicates that it depends on the data whether additional freedom in the parameter space is a benefit or hindrance \emph{for the algorithm} to find a good solution.  More importantly, as we will see in the next section, starting with HyCom as initialization, our model is always significantly better.

\begin{table}[t]
	\caption{Test statistic of the likelihood ratio test between our models and block models, and our models and HyCom. For all datasets 500 communities were sampled. Additionally, we display the result for sampling 100 communities from the \dblp data.}
	\label{tab:LLratioTest}
	\centering
	\begin{tabular}{@{}lrrrr@{}}
		\toprule
		& \multicolumn{2}{c}{LL ratio} \\ 
		\cmidrule(rl){2-3}
				& block model & HyCom\\  
		\midrule
		\amazon &  26450.6& 30997.1 \\ 
		\dblp(100) & 3148.5 & -788.0\\ 
		\dblp & -264974.7 & 17958.1\\ 
		\friendster &  200627.6 & 17811.7\\ 
		\lj &   154982.4 & 22705.8\\ 
		\orkut &  11945.3 & 1598.5\\ 
		\youtube &   75689.6 & 12660.0\\ 
		\bottomrule 
	\end{tabular}
	\end{table}

\subsection{Finding communities}
\label{sec:experiments:findingCommunities}

Next, we demonstrate that our model improves the description of communities returned by existing community-finding methods. We used spectral clustering, Boolean matrix factorization, and HyCom to find the communities; the first two approaches look for clique-like communities while HyCom looks for hyperbolic shapes (see Section~\ref{sec:hycom}). We used various real-world data sets from the  University of Florida Sparse Matrix Collection~\cite{davis2011university}, summarized in Table~\ref{tab:datasets_small}. They are significantly smaller than the data sets we examined in the previous experiments to allow the community detection algorithms to work efficiently. As we aim to assess the quality of the models and do not propose an own method for finding communities, no ground-truth community information was employed for the evaluation.

\begin{table}[t]
	\caption{Datasets used for the community finding experiments.}
	\label{tab:datasets_small}
	\centering
	\begin{tabular}{@{}l@{\hspace{0.4em}}r@{\hspace{0.4em}}rl@{}}
		\toprule
		& nodes &  non-zeros & content\\ 
		\midrule
		\emaild & 1,133 & 10,902 & University email network\\ 
		\erdos & 472& 2,628	 & Erd\H{o}s collaboration network\\ 
		\jazz & 198 & 5,484.6& Network of Jazz musicians \\ 
		\polbooks &  105 & 882 & Books about US politics \\ 
		\bottomrule 
	\end{tabular}
\end{table}
 
\paragraph{Spectral clustering}
We first used spectral clustering with the normalized Laplacian~\cite{luxburg07tutorial} to cluster the nodes of the graph. The resulting communities are non-overlapping. We notice a significant benefit of modelling the obtained result by means of hyperbolic models, as the log-likelihood ratio test confirms for all examined datasets (see Table~\ref{tab:findingCommunitiesResults}). These results yielded $p$-values that were essentially zero, confirming that the results are statistically significant.  We have used $k=10$ clusters for \emaild, $k=8$ for \erdos, $k=5$ for \jazz, and $k=6$ for \polbooks. We display examples of modelled communities in Figure~\ref{fig:example_spectClust}. These example communities show relatively large cores but thin tails, with most edges being in the lower triangular area. Our models clearly capture this phenomenon.
 
\begin{figure}[tb]
	\centering
	\subfloat[Community from \emaild data.]{
		\centering
		\includegraphics[width=0.472 \linewidth]{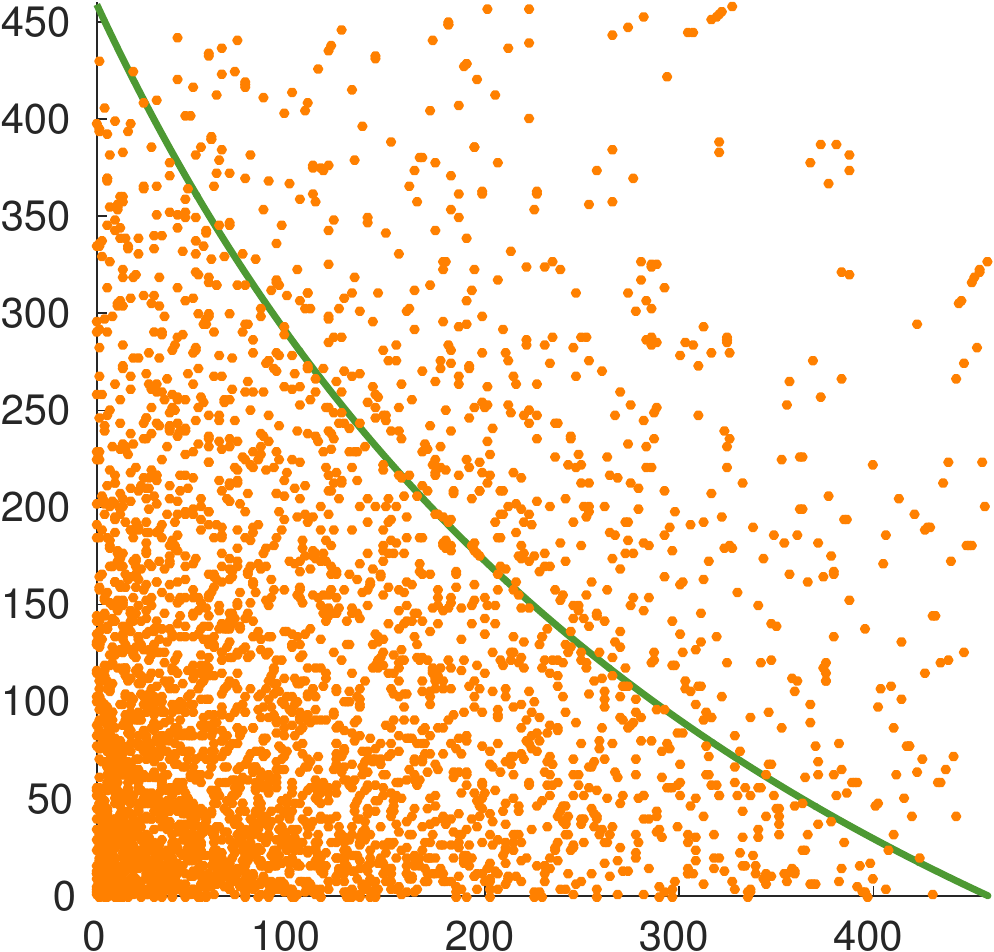}
		\label{fig:email_spectClust_2}
	}
	\subfloat[Community from \jazz data.]{
		\centering
		\includegraphics[width=0.472 \linewidth]{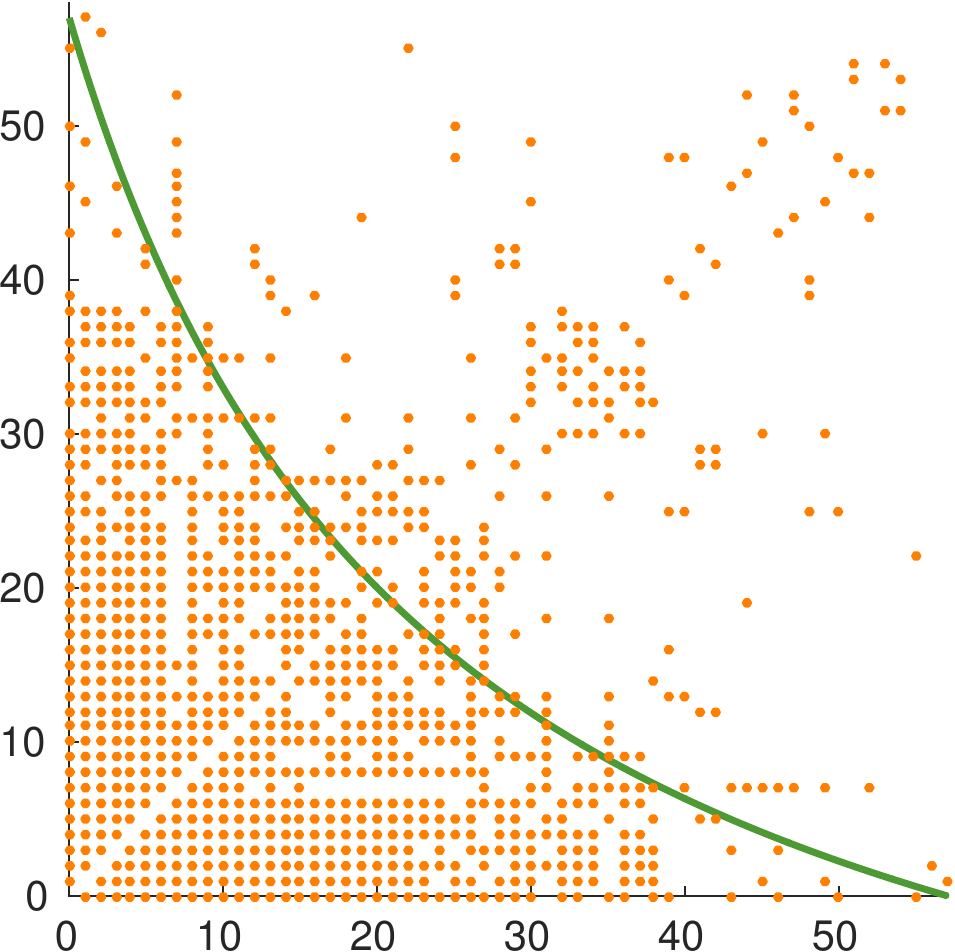}
		\label{fig:jazz_spectClust_2}
	}
		\caption{Examples of communities fitted by our model. The initial communities were obtained using spectral clustering on the respective data.}
	\label{fig:example_spectClust}
\end{figure}
\begin{figure}[tb]
	\centering
	\subfloat[Community from \jazz data.]{
		\centering
		\includegraphics[width=0.472 \linewidth]{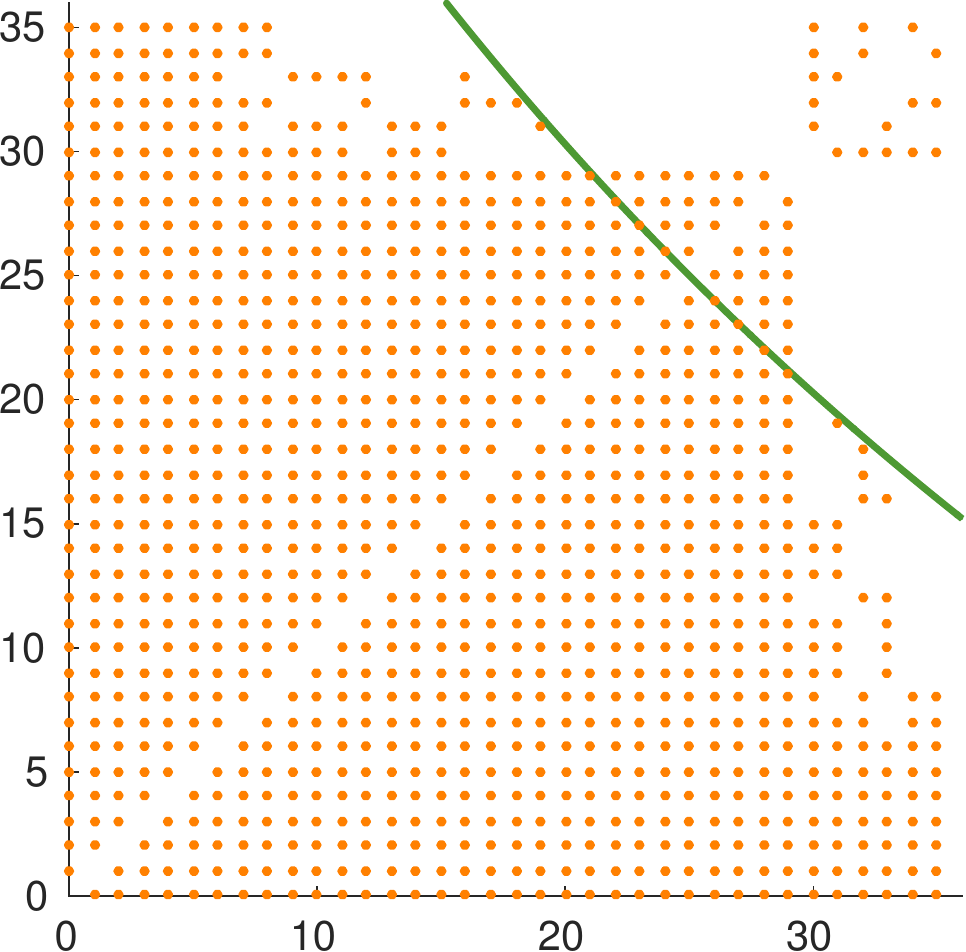}
		\label{fig:jazz_BMF_2}
	}
	\subfloat[Community from \polbooks data.]{
		\centering
		\includegraphics[width=0.472 \linewidth]{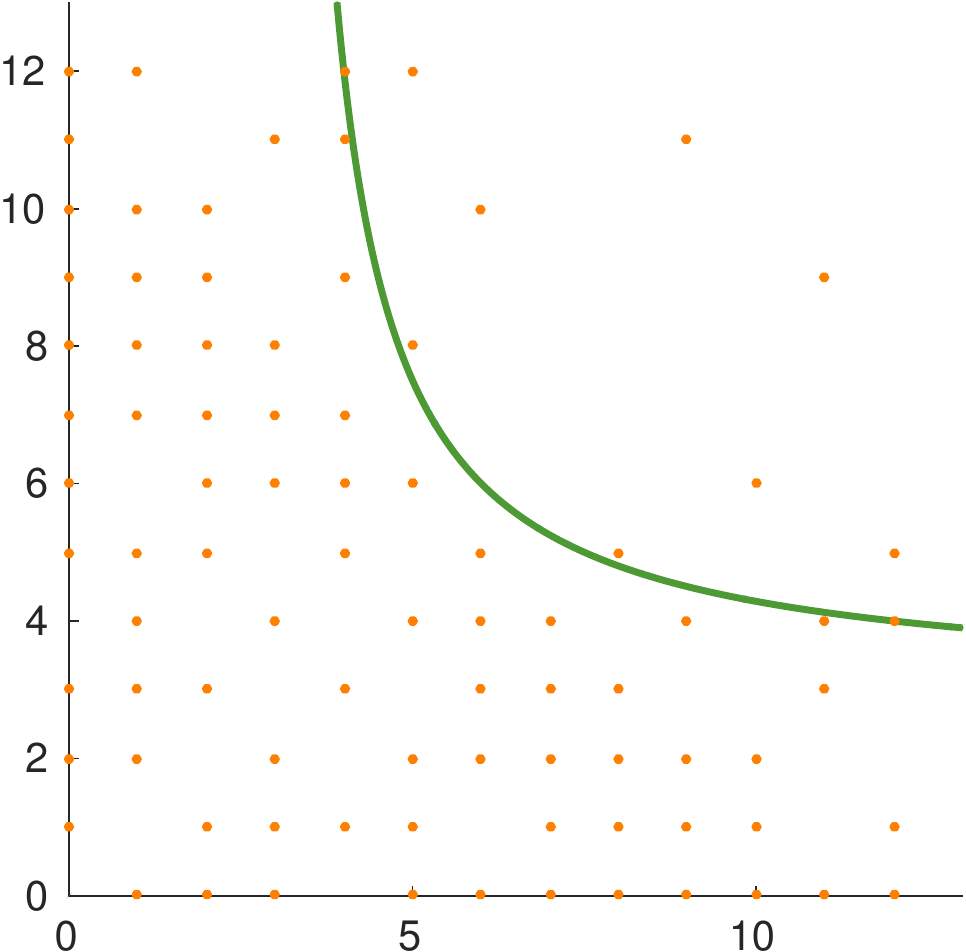}
		\label{fig:polbooks_BMF_2}
	}
		\caption{
			Examples of communities fitted by our model. The initial communities were obtained using Boolean matrix factorization.
		}
	\label{fig:example_BMF}
\end{figure}

\paragraph{Boolean matrix factorization}
To find overlapping communities, we used Boolean matrix factorization (BMF)~\cite{miettinen2008discrete}. We used the Asso algorithm~\cite{miettinen2008discrete} with the same number of communities as with spectral clustering. We set the threshold parameter $\tau$ of Asso to $0.6$ and the weight $w$ to $10$. We again find that our models resemble the data significantly better than the corresponding block models (see Table~\ref{tab:findingCommunitiesResults}). Figure~\ref{fig:example_BMF} shows example communities from the \polbooks and \jazz data. 

\paragraph{HyCom algorithm}
We ran the HyCom algorithm~\cite{araujo2014beyond} on each of the data sets stopping after we found $k$ communities, with $k$ specified as above. As the likelihood ratio test confirms (Table~\ref{tab:findingCommunitiesResults}), our hyperbolic model improves the result of the HyCom algorithm. 

\begin{table}[t]
	\caption{Statistic of the likelihood ratio test between our and block models, as well as our and HyCom models. The communities are found using HyCom, spectral clustering, and BMF.}
	\label{tab:findingCommunitiesResults}
	\centering
	\begin{tabular}{@{}lrrr@{}}
		\toprule
		& \multicolumn{3}{c}{LL ratio} \\ 
		\cmidrule(rl){2-4}
		& \multicolumn{1}{c}{spectral clustering} & \multicolumn{1}{c}{BMF} & \multicolumn{1}{c}{HyCom}\\ 
		\midrule
		\emaild & 10895.8 & 3552.0 & 250.1\\
		\erdos & 1797.0 & 949.0 & 256.3\\
		\jazz & 3003.8 & 4435.0 & 3718.5\\
		\polbooks & 648.0 & 303.3 & 228.2\\
		\bottomrule 
	\end{tabular}
\end{table}


\subsection{Discussion}
\label{sec:discussion}

Our experiments have verified our intuition that the communities in real-world graphs are better modelled using our models than the traditional quasi-clique models, and that our models are an improvement over the previously proposed HyCom model~\cite{araujo2014beyond}. This holds true for a variety of data sets, both with ground-truth communities, and with communities detected with existing methods. It is important to notice that the existing methods, especially the BMF, aim at finding clique-like communities. Thus, since our algorithm uses their results as the initial community candidates, any weaknesses of these algorithms will also affect our result. Still, our experiments show that our models provide statistically significantly better fit, even when we take into account the increased number of free parameters for our model. 

Not only is our model a better fit for the data, it also provides interesting insights to the shape of the communities. The easy interpretability of the parameters $\gamma$ and $H$ means that we can simply study a summary of their distributions to gain an understanding on how the communities in a data look like, whether the cores are small or big and whether the tails are fat or skinny. This allows a data analyst to obtain a fast general understanding about the data without having to look at any particular community.

Finally, our experiments also demonstrate the scalability of our method. It had no problem of handling even the largest graph, \friendster, with approximately $65.6$ million nodes and $1.8$ billion edges.


\section{Conclusions}
\label{sec:conclusions}

We have proposed three novel models to describe hyperbolic communities. Based on the observation that communities in real-world graphs do not correspond to blocks of uniform density, our models capture the density distributions per community more accurately.
We have shown that all models have the same expressive power, and we proposed an algorithm to fit these models to a given community -- and likewise, to fit multiple, potentially overlapping, communities to represent the full graph.

In our experimental study, we have analysed a large variety of real-world datasets leading to interesting insights about the data's inherent community structure -- showing variations from star-like data to data with patterns similar to quasi-cliques. Our hyperbolic model captures all these scenarios as special cases. 
Last, comparing the likelihood obtained by our model w.r.t.\ block-modelling approaches and existing hyperbolic models clearly shows the superiority of our hyperbolic community model for real-world data.

\paragraph{Future work}
While our current model allows node overlapping communities, we have restricted edges to be part of at most one community. Extending our models and algorithms to handle edge overlaps is an important research direction we aim to investigate.
Moreover, while this work focused on modeling a set of communities, we aim to investigate community detection algorithms able to detect hyperbolic structures directly. To that end, the aforementioned nonnegative rounding rank decompositions~\cite{neumann16what} provide an interesting approach.


\section*{Acknowledgements}
\small This research was supported by the German Research Foundation (DFG), Emmy Noether grant GU 1409/2-1, and by the Technical University of Munich - Institute for Advanced Study, funded by the German Excellence Initiative and
the European Union Seventh Framework Programme under grant agreement
no 291763, co-funded by the European Union.

\bibliographystyle{abbrv}
\bibliography{bibliography}

\appendix

\section{Proof of Proposition 4}
\label{sec:proof:prop:4}

Recall that the Proposition 4 was:
\setcounter{theorem}{3}
\begin{proposition}
  \label{prop:area_with_varying_gamma} Let $C = \mfixed{\gamma, H}$ be a community of $n_C$ nodes defined by $\gamma\in\R$ and $H\in\N$, and let $A_C$ be its area. Then there exists integer $\gamma'$ such that if $D = \mfixed{\gamma', H}$ is the community defined by $\gamma'$ and $H$, and $A_D$ is the respective area, then $\abs{A_C - A_D} \in \Theta(\gamma \ln(n_C))$.
\end{proposition}

We will prove the claim by bounding the difference of the area between $\gamma$ and $\gamma+1$, which is clearly sufficient. 

Notice first that we can assume $H=0$: the tail part will always contribute the same area of $Hn_C - H^2/2$ irrespective of $\gamma$. 

Consider for now the \mhyperbolic{} model and recall that the hyperbolic equation is
\begin{equation}
  \label{eq:1}
  (x + p)(y + p) = \theta\; ,
\end{equation}
where we used $x$ and $y$ instead of $i$ and $j$ to emphasize their continuous nature. From \eqref{eq:1} we get that 
\begin{equation}
  \label{eq:2}
  y = \frac{\theta}{x+p} - p\; .
\end{equation}
The area of the function is the integral of \eqref{eq:2} from 0 to $n_C$, that is,
\begin{equation}
  \label{eq:3}
  A_C = \int_0^{n_C}\frac{\theta}{x+p} - p\; \mathrm{d} x\; .
\end{equation}
Here we dropped the $-1$ from $n_C$ for the sake of clarity; it will not effect the asymptotic analysis in any case. To get back to the \mfixed{} model, we can substitute
\begin{align}
  \label{eq:5}
  p  &= \frac{\gamma^2}{n_C - 2\gamma} \\
  \theta &= \frac{\gamma^2(\gamma^2 - n_C)^2}{(n_C - 2\gamma)^2} \; , \label{eq:theta}
\end{align}
following equations (8) and (9) of the main paper, respectively. 

Given the constraints that $\gamma \in [0,n_C/2)$, we can solve the integral \eqref{eq:3} with the substitutions \eqref{eq:5} and \eqref{eq:theta} to be
\begin{multline}
  \label{eq:4}
  \bigl({\gamma}^{2} ( 4 \gamma \ln( \gamma )  n_C 
-2 {\gamma}^{2}\ln( \gamma ) -2 \ln( \gamma ) { n_C }^{2} 
+2 {\gamma}^{2}\ln( -\gamma+ n_C  ) \\
-{ n_C }^{2}+2 \gamma  n_C +2 \ln( -\gamma+ n_C  ) { n_C }^{2}
-4 \gamma\ln( -\gamma+ n_C  )  n_C  )\bigr)\\
/\bigl({ n_C }^{2}-4 \gamma  n_C +4 {\gamma}^{2}\bigr) \; .
\end{multline}

Consequently, the difference between the areas with $\gamma$ and $\gamma+1$ is 
\begin{multline}
\label{eq:big}
-\bigl(-2\, n_C +\ln   ( \gamma^2+2\,\gamma +1  ) \gamma^2 n_C^2\\
-\ln   ( \gamma^2+2\,\gamma -2\,\gamma  n_C +1-2\, n_C + n_C^2  ) \gamma^2 n_C^2\\
 -6\,\gamma  n_C +\ln   ( \gamma^2+2\,\gamma +1  )  n_C^2\\
+\gamma^4\ln   ( \gamma^2+2\,\gamma +1  ) -\ln   ( \gamma^2+2\,\gamma -2\,\gamma  n_C +1-2\, n_C + n_C^2  )  n_C^2\\
 -\gamma^4\ln   ( \gamma^2+2\,\gamma -2\,\gamma  n_C +1-2\, n_C + n_C^2  ) \\
+\gamma^2 n_C^2-4\,\gamma^3\ln   ( \gamma^2+2\,\gamma -2\,\gamma  n_C +1-2\, n_C + n_C^2  ) \\
-6\,\ln   ( \gamma^2+2\,\gamma -2\,\gamma  n_C +1-2\, n_C + n_C^2  ) \gamma^2\\
 -4\,\ln   ( \gamma^2+2\,\gamma -2\,\gamma  n_C +1-2\, n_C + n_C^2  ) \gamma \\
+4\,\gamma^3\ln   ( \gamma^2+2\,\gamma +1  ) +6\,\ln   ( \gamma^2+2\,\gamma +1  ) \gamma^2\\
 +4\,\ln   ( \gamma^2+2\,\gamma +1  ) \gamma -2\,\ln   ( \gamma^2+2\,\gamma +1  )  n_C \\
+2\,\ln   ( \gamma^2+2\,\gamma -2\,\gamma  n_C +1-2\, n_C + n_C^2  )  n_C \\
 -2\,\gamma^3 n_C -6\,\gamma^2 n_C +2\,\gamma  n_C^2+ n_C^2\\
-\ln   ( \gamma^2+2\,\gamma -2\,\gamma  n_C +1-2\, n_C + n_C^2  ) \\
+\ln   ( \gamma^2+2\,\gamma +1  ) -2\,\gamma^3\ln   ( \gamma^2+2\,\gamma +1  )  n_C \\
 +6\,\ln   ( \gamma^2+2\,\gamma -2\,\gamma  n_C +1-2\, n_C + n_C^2  ) \gamma^2 n_C \\
+6\,\ln   ( \gamma^2+2\,\gamma -2\,\gamma  n_C +1-2\, n_C + n_C^2  ) \gamma  n_C \\
 -6\,\ln   ( \gamma^2+2\,\gamma +1  ) \gamma^2 n_C \\
+2\,\ln   ( \gamma^2+2\,\gamma +1  ) \gamma  n_C^2\\
-6\,\ln   ( \gamma^2+2\,\gamma +1  ) \gamma  n_C \\
 -2\,\ln   ( \gamma^2+2\,\gamma -2\,\gamma  n_C +1-2\, n_C + n_C^2  ) \gamma  n_C^2\\
+2\,\gamma^3\ln   ( \gamma^2+2\,\gamma -2\,\gamma  n_C +1-2\, n_C + n_C^2  )  n_C \bigr)\\
/\bigl( n_C^2-4\,\gamma  n_C -4\, n_C +4\,\gamma^2+8\,\gamma +4\bigr)\\
+\bigl( \gamma^2  ( -4\,\gamma \ln   ( \gamma   )  n_C +2\,\gamma^2\ln   ( \gamma   ) \\
+2\,\ln   ( \gamma   )  n_C^2-\gamma^2\ln   ( \gamma^2-2\,\gamma  n_C + n_C^2  ) \\
+ n_C^2-2\,\gamma  n_C -\ln   ( \gamma^2-2\,\gamma  n_C + n_C^2  )  n_C^2\\
 +2\,\gamma \ln   ( \gamma^2-2\,\gamma  n_C + n_C^2  )  n_C   )\bigr)\\
/\bigl( n_C^2-4\,\gamma  n_C +4\,\gamma^2\bigr)
\end{multline}

If we denote \eqref{eq:big} with $d(\gamma, n_C)$, we can notice that
\begin{equation}
  \label{eq:6}
  \lim_{n_C\to\infty}\frac{\gamma\ln(n_C)}{d(\gamma, n_C)} \to \frac{1}{2/\gamma + 4}\; ,
\end{equation}
which is in $(0, 1/4)$ for any $\gamma$, showing that $d(\gamma, n_C)\in \Theta(\gamma\ln(n_C))$ and concluding the proof. 

\section{Examples of Modelled Communities}
\subsection{Models from Ground-Truth Communities}
Figures~\ref{fig:dblp_example}, \ref{fig:friendster_example}, and \ref{fig:youtube_example} show further examples of communities from different data sets and the models we found.
\subsection{Models after Spectral Clustering}
Figures~\ref{fig:erdos_spectClust}, \ref{fig:jazz_spectClust}, and \ref{fig:email_spectClust} give additional examples of the modelled communities obtained by spectral clustering.
\subsection{Models after Boolean Matrix Factorization}
Figures~\ref{fig:email_BMF}, \ref{fig:erdos_BMF}, \ref{fig:polbooks_BMF}, and \ref{fig:jazz_BMF}, show the models obtained for communities found by Boolean matrix factorization.

\begin{figure}[h]
	\centering
	\subfloat[Before.]{
		\centering
		\includegraphics[width=0.472 \linewidth]{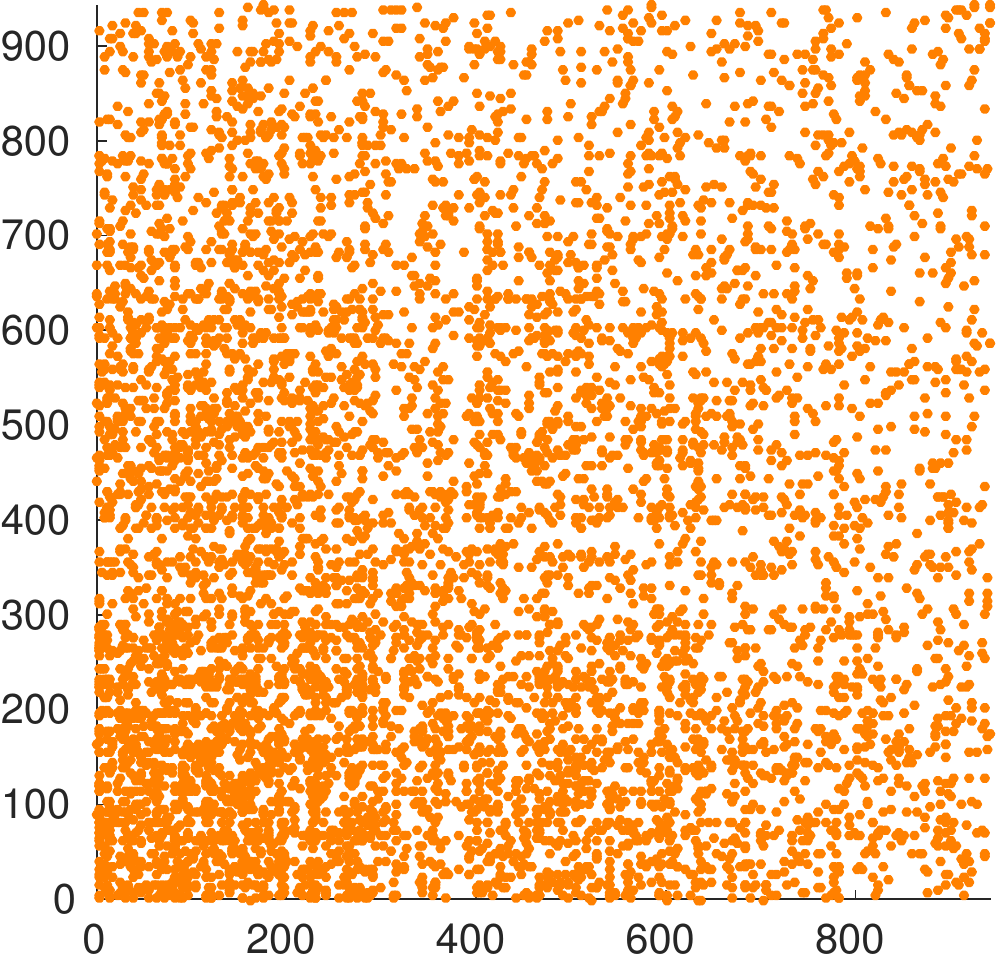}
		\label{fig:dblp_before}
	}%
	\subfloat[After fitting the model.]{
		\centering
		\includegraphics[width=0.472 \linewidth]{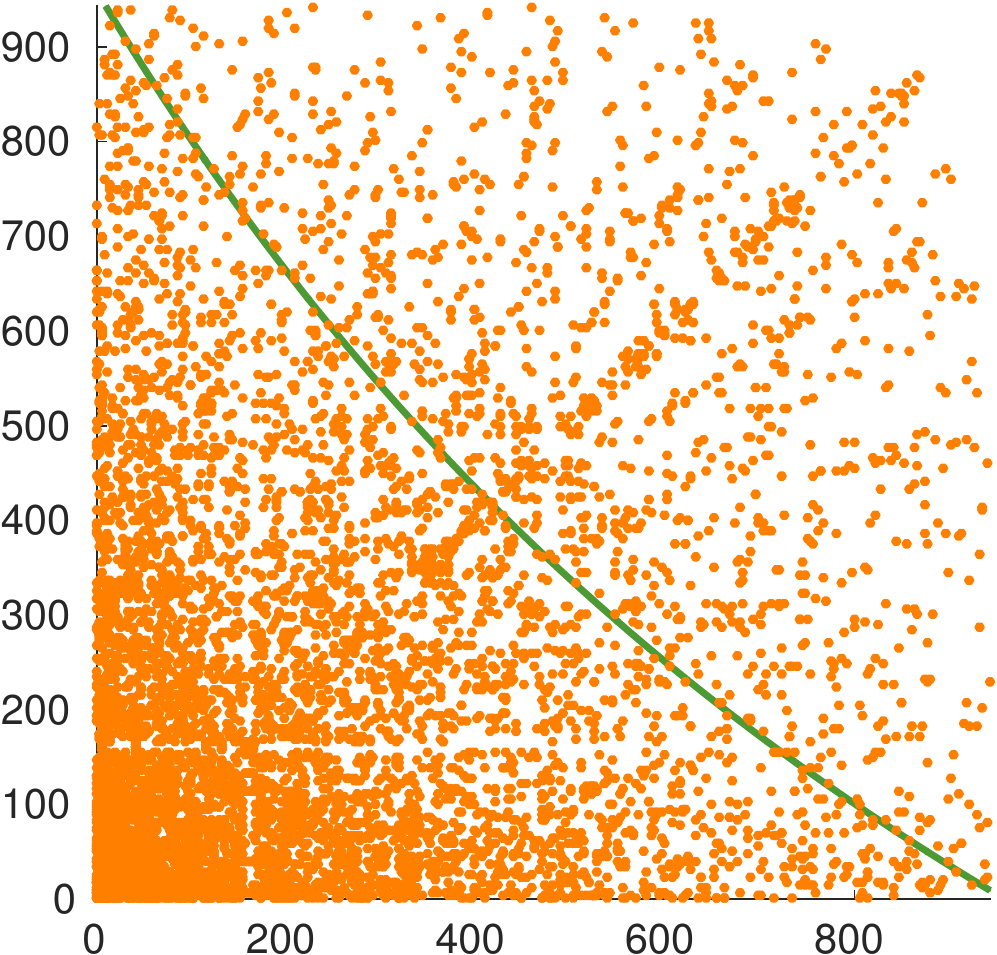}
		\label{fig:dblp_after}
	}
	\caption{Example of a community obtained from \dblp data.}
	\label{fig:dblp_example}
\end{figure}
\begin{figure}[!h]
	\centering
	\subfloat[Before.]{
		\centering
		\includegraphics[width=0.472 \linewidth]{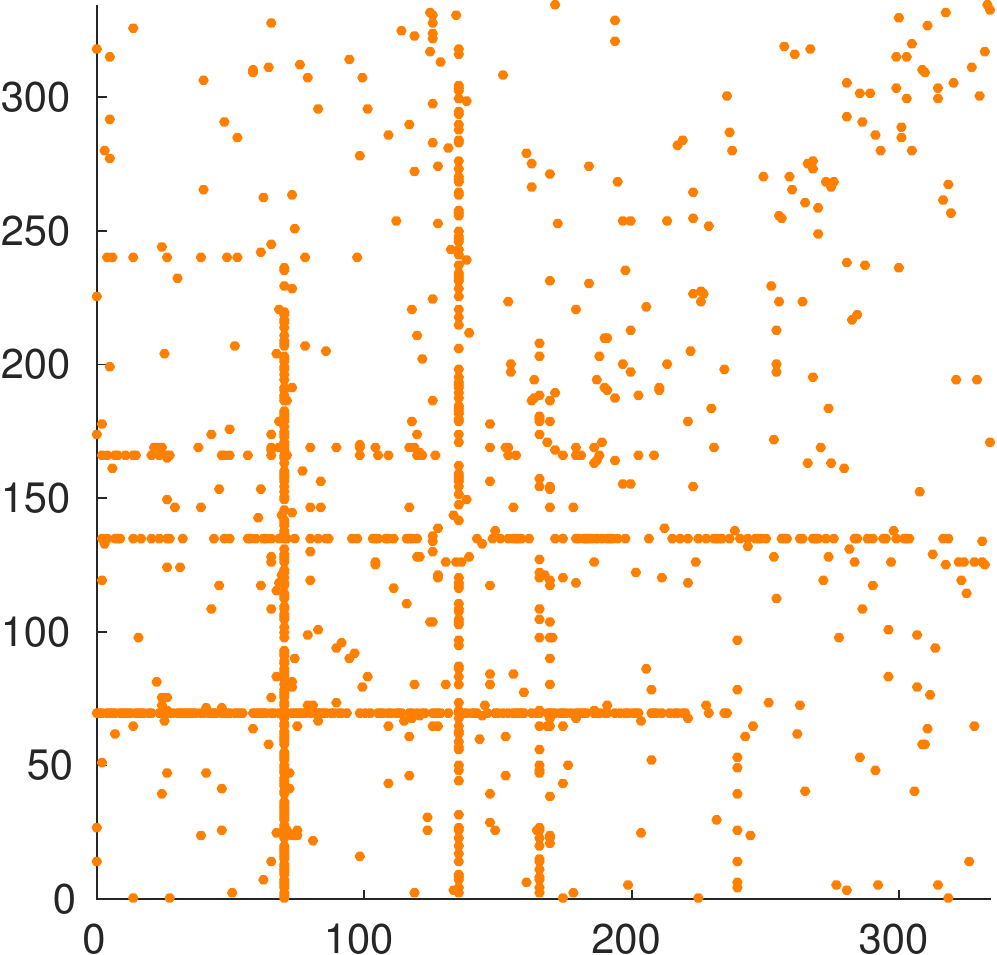}
		\label{fig:friendster_before}
	}%
	\subfloat[After fitting the model.]{
		\centering
		\includegraphics[width=0.472 \linewidth]{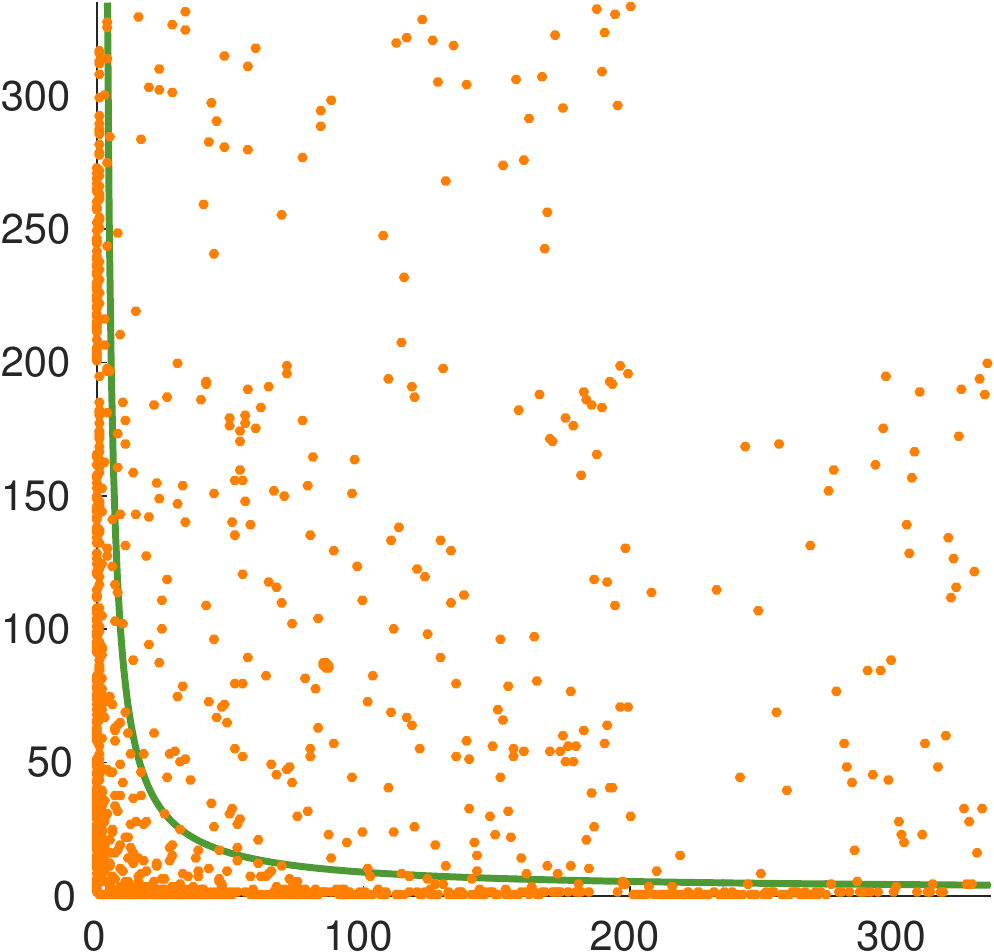}
		\label{fig:friendster_after}
	}
	\caption{Example of a community obtained from \friendster data.}
	\label{fig:friendster_example}
\end{figure}
\begin{figure}[!h]
	\centering
	\subfloat[Before.]{
		\centering
		\includegraphics[width=0.472 \linewidth]{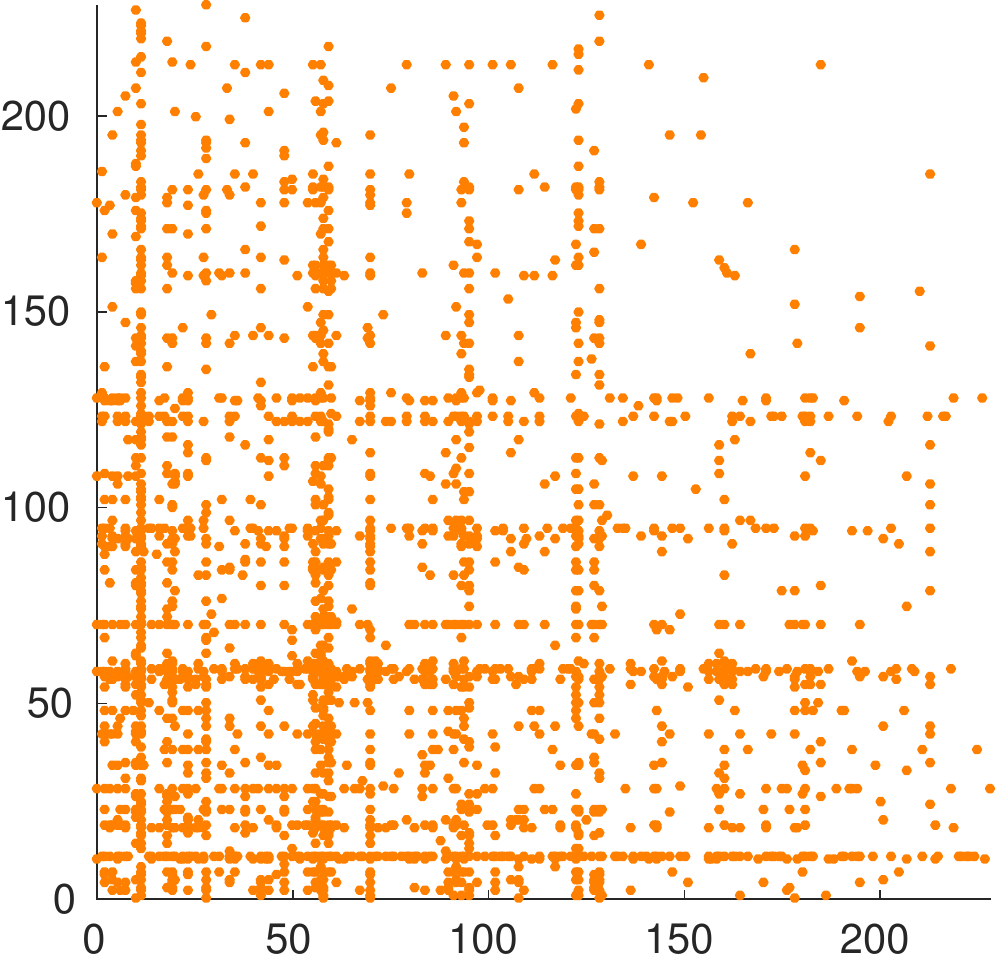}
		\label{fig:youtube_before}
	}%
	\subfloat[After fitting the model.]{
		\centering
		\includegraphics[width=0.472 \linewidth]{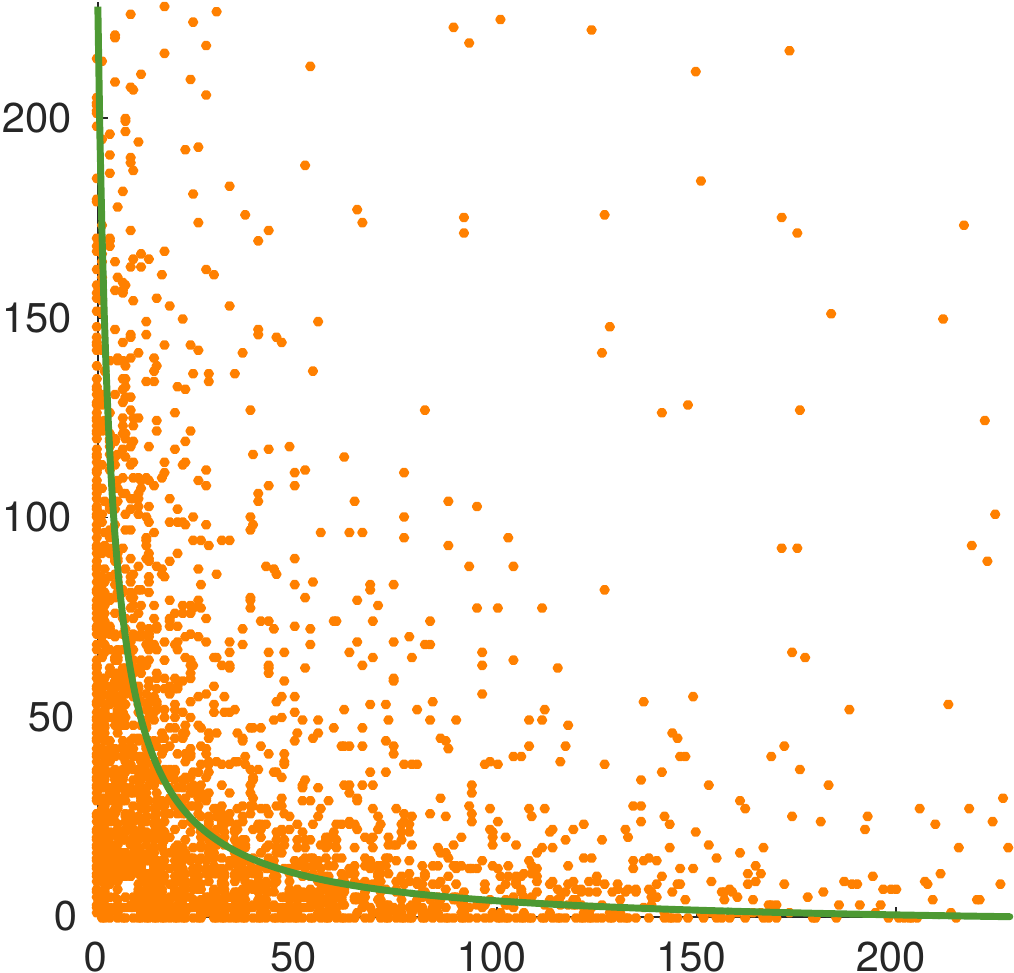}
		\label{fig:youtuber_after}
	}
	\caption{Example of a community obtained from \youtube data.}
	\label{fig:youtube_example}
\end{figure}

\begin{figure}[!h]
	\centering
	\subfloat[Before.]{
		\centering
		\includegraphics[width=0.472 \linewidth]{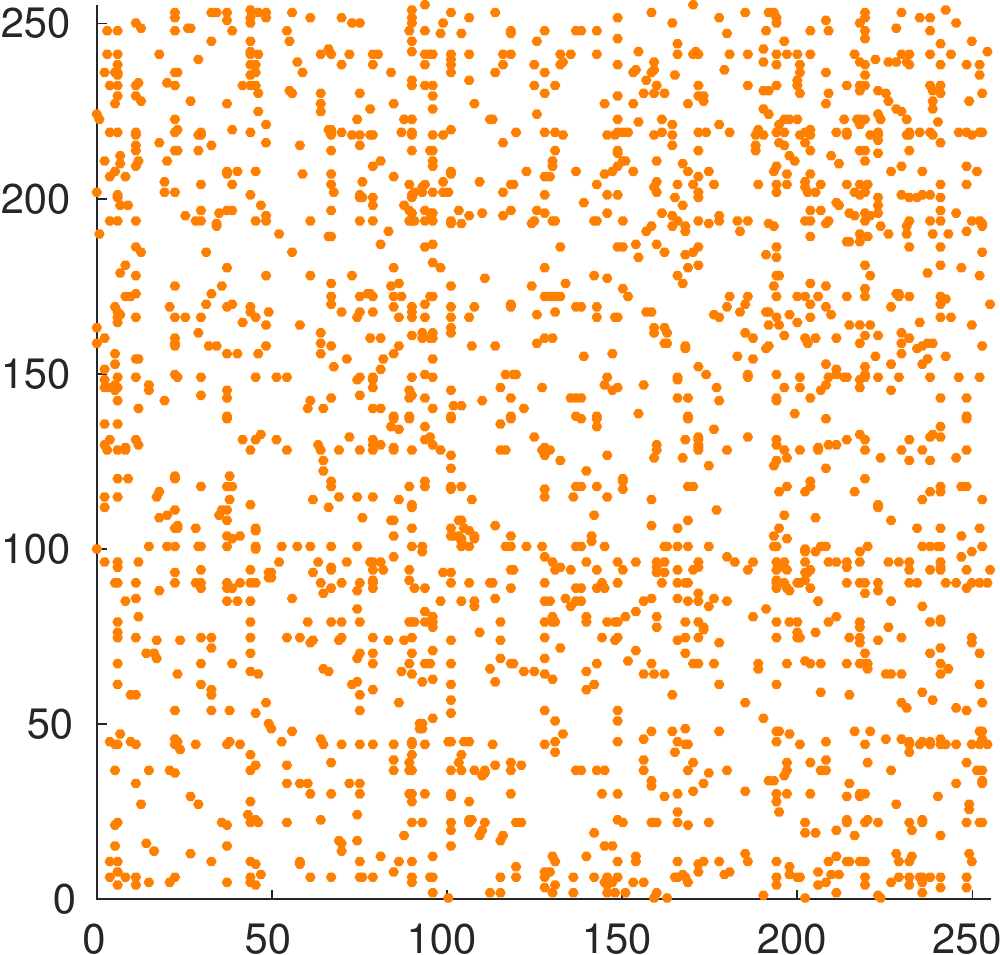}
		\label{fig:erdos_spectClust_before}
	}%
	\subfloat[After fitting the model.]{
		\centering
		\includegraphics[width=0.472 \linewidth]{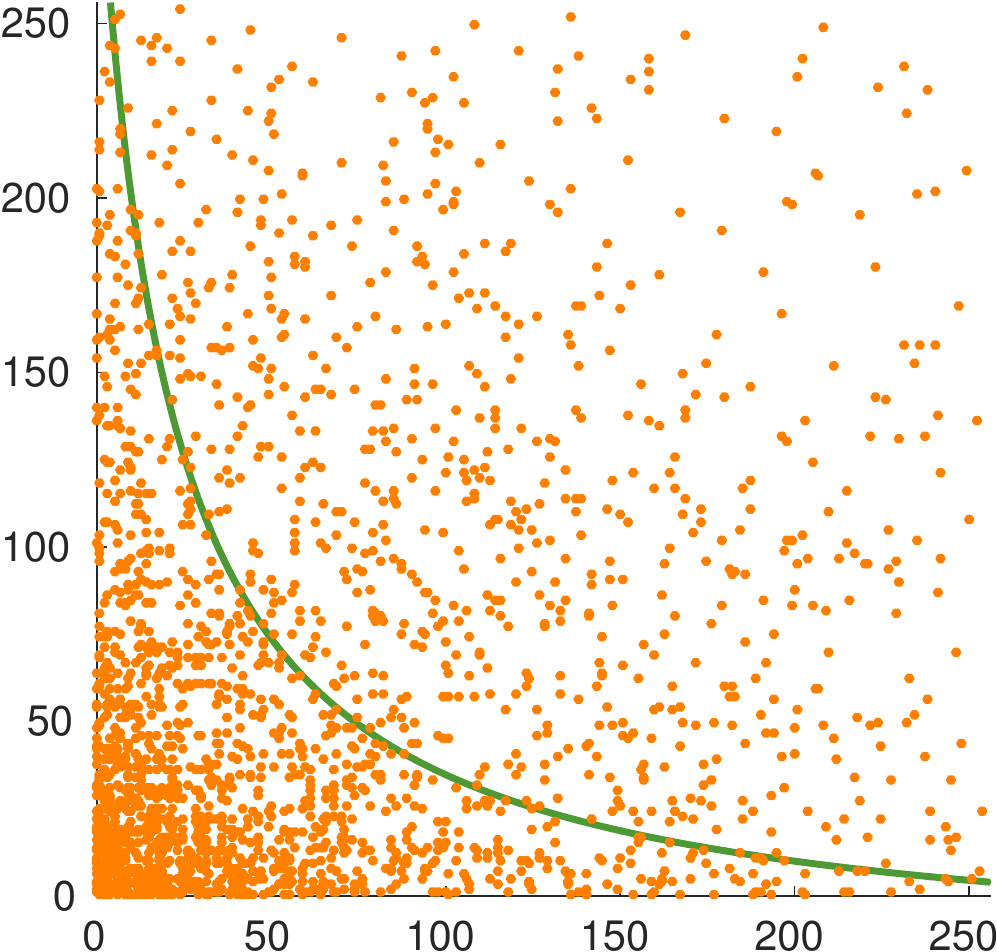}
		\label{fig:erdos_spectClust_after}
	}
	\caption{Examples of a community obtained from spectral clustering of the \erdos data.}
	\label{fig:erdos_spectClust}
\end{figure}
\begin{figure}[!h]
	\centering
	\subfloat[]{
		\centering
		\includegraphics[width=0.472 \linewidth]{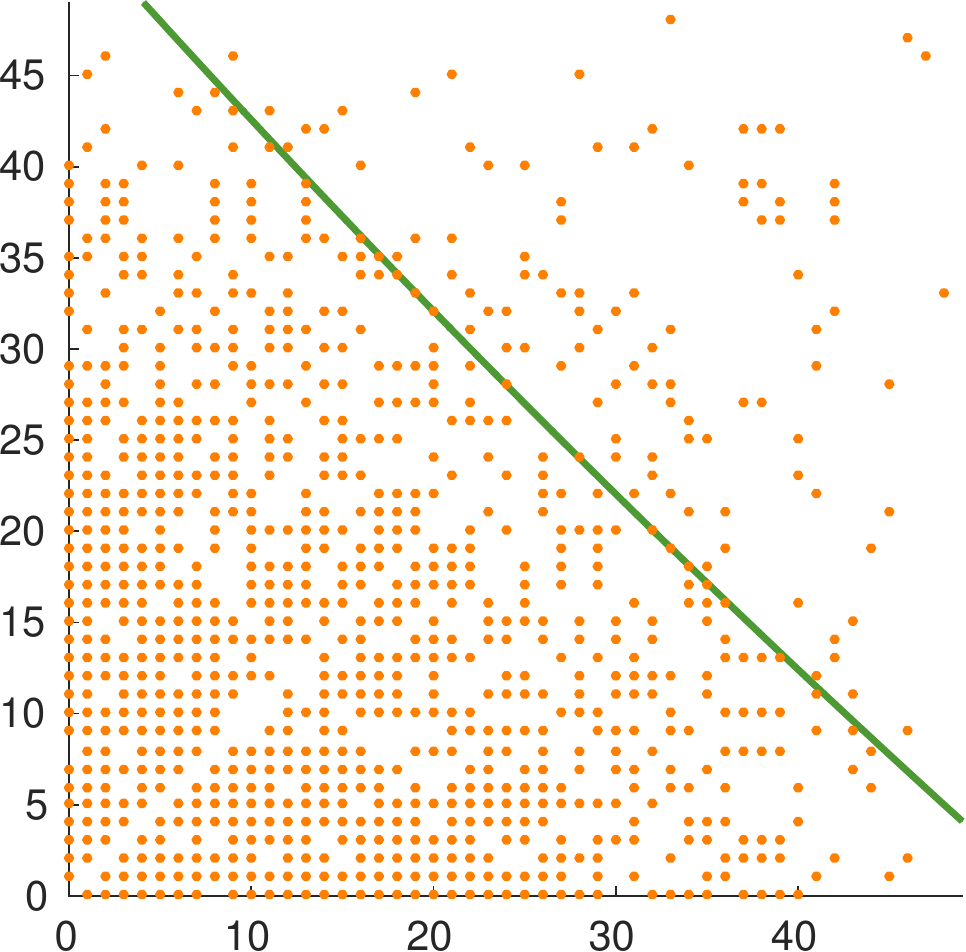}
		\label{fig:jazz_spectClust_1}
	}%
	\subfloat[]{
		\centering
		\includegraphics[width=0.472 \linewidth]{jazz_spectClust_2}
		\label{fig:jazz_spectClust_2}
	}%
		\caption{Examples of communities obtained from spectral clustering of the \jazz data fitted by our model.}
	\label{fig:jazz_spectClust}
\end{figure}
\begin{figure}[!h]
	\centering
	\subfloat[]{
		\centering
		\includegraphics[width=0.472 \linewidth]{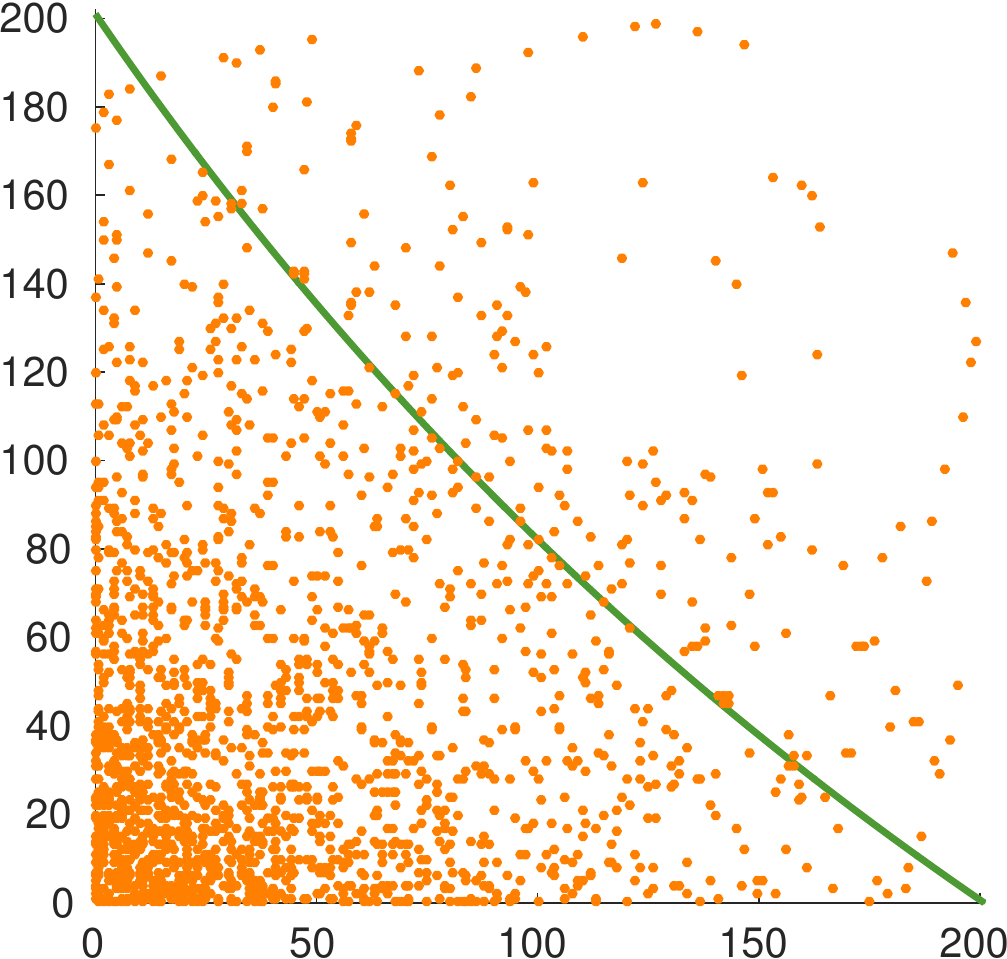}
		\label{fig:email_spectClust_1}
	}%
		\subfloat[]{
		\centering
		\includegraphics[width=0.472 \linewidth]{email_spectClust_2}
		\label{fig:email_spectClust_2}
	}%
	\caption{Examples of communities obtained from spectral clustering of the \emaild data fitted by our model.}
	\label{fig:email_spectClust}
\end{figure}

\begin{figure}[!h]
	\centering
	\subfloat[]{
		\centering
		\includegraphics[width=0.472 \linewidth]{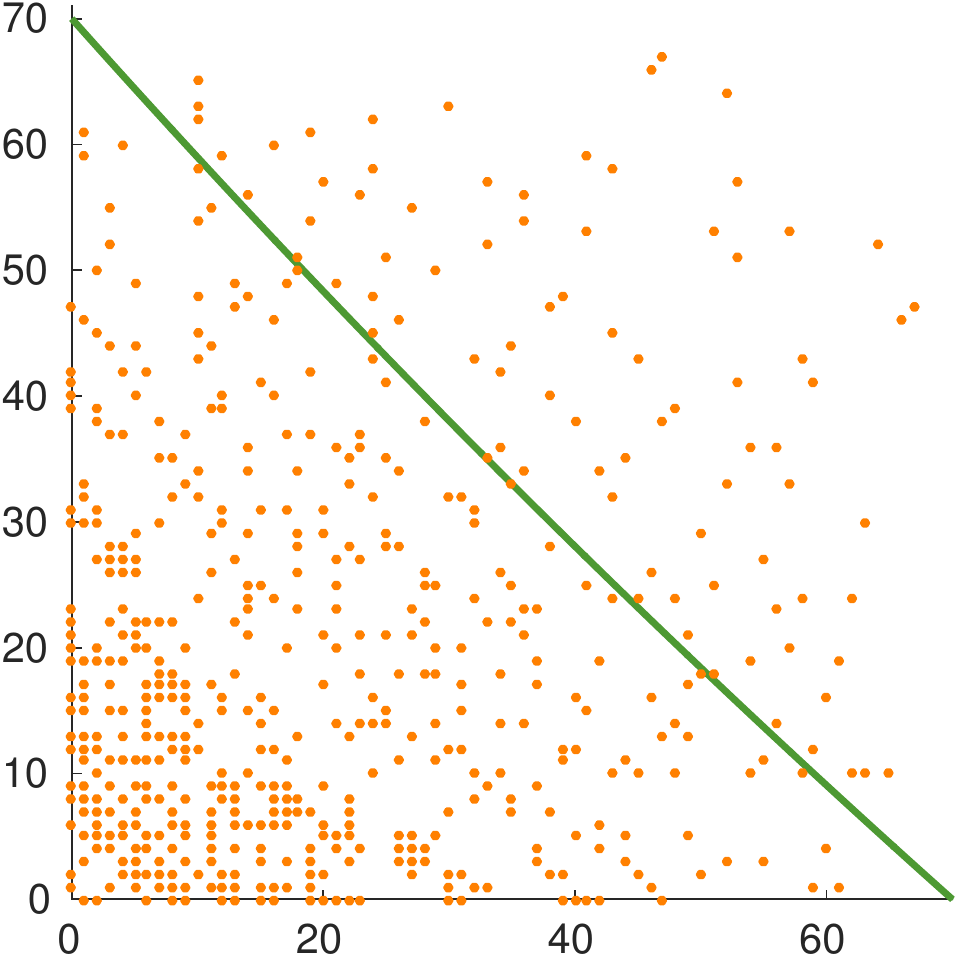}
		\label{fig:email_BMF_1}
	}%
	\subfloat[]{
		\centering
		\includegraphics[width=0.472 \linewidth]{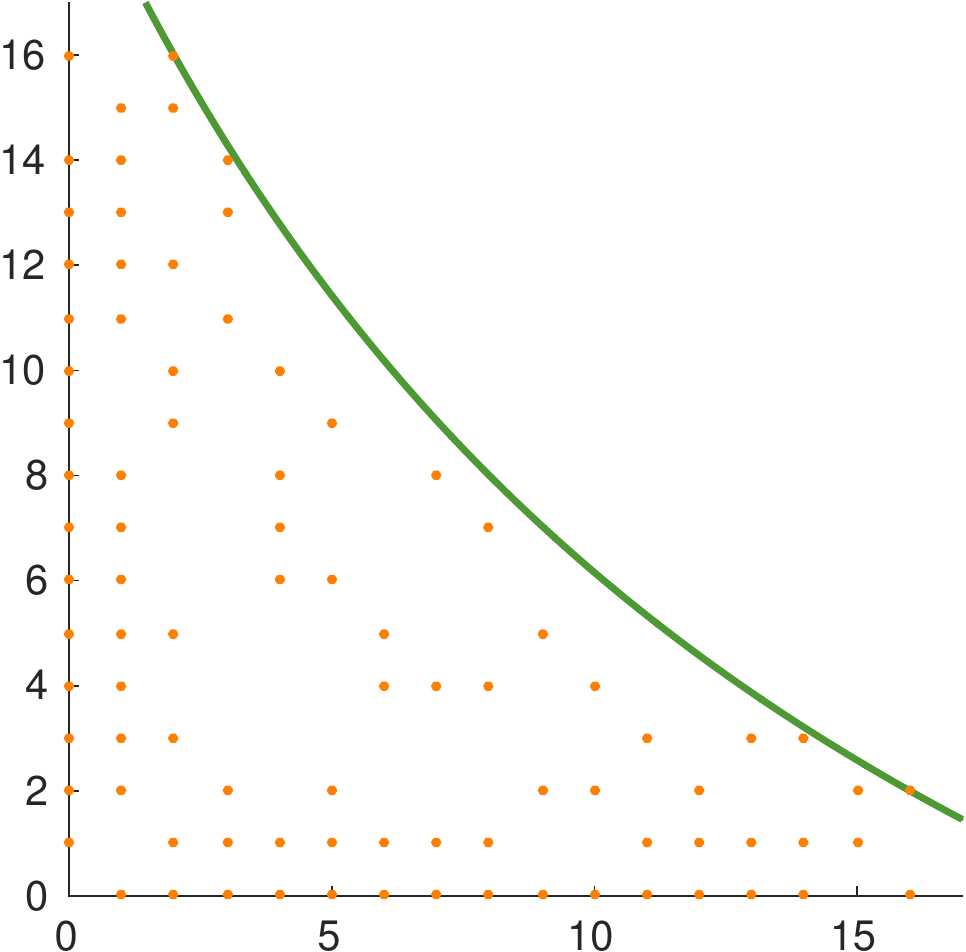}
		\label{fig:email_BMF_2}
	}
	\caption{Examples of communities obtained from Boolean matrix factorization of the \emaild data.}
	\label{fig:email_BMF}
\end{figure}
\begin{figure}[!h]
	\centering
	\subfloat[]{
		\centering
		\includegraphics[width=0.472 \linewidth]{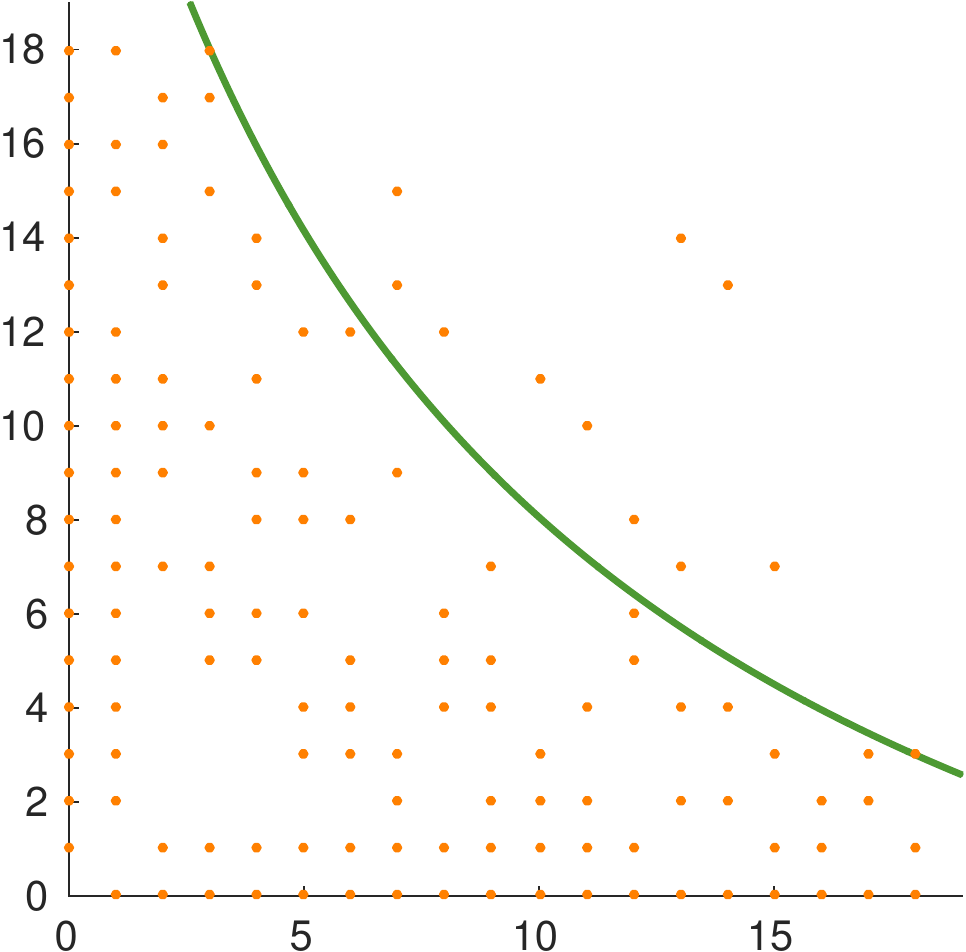}
		\label{fig:erdos_BMF_1}
	}%
	\subfloat[]{
		\centering
		\includegraphics[width=0.472 \linewidth]{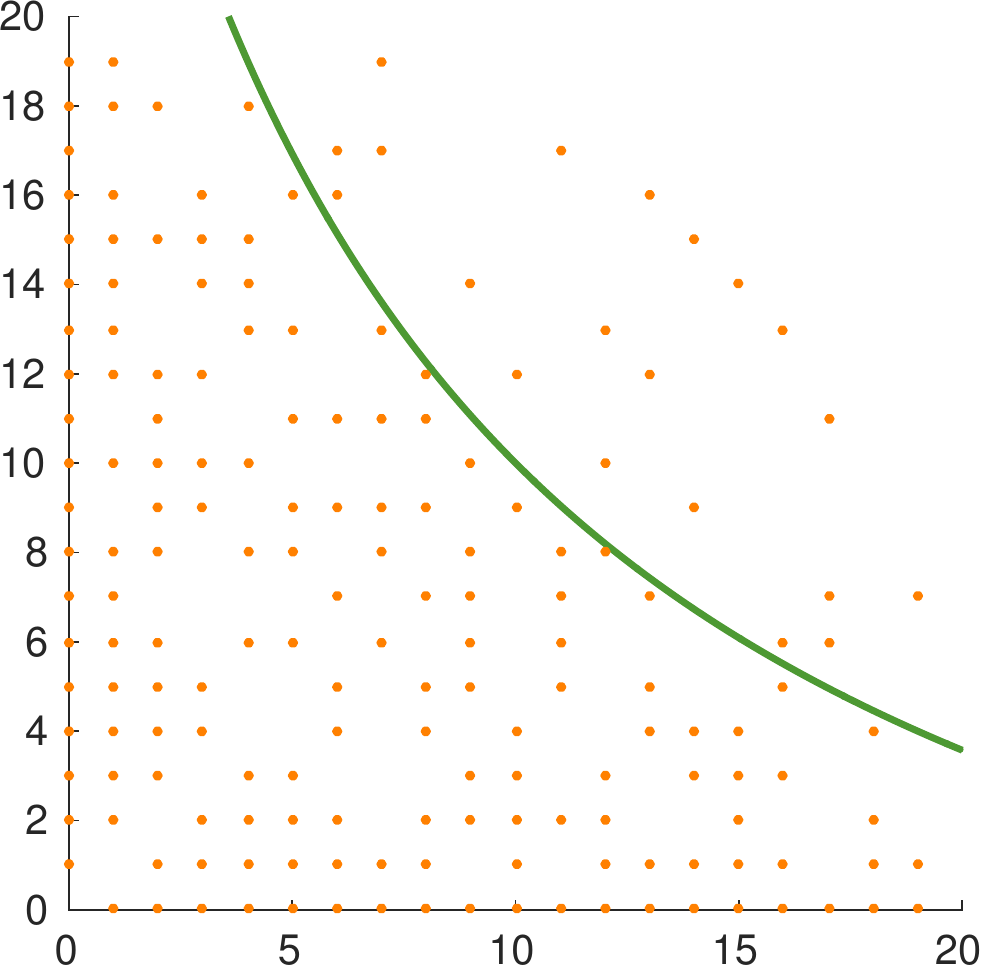}
		\label{fig:erdos_BMF_2}
	}
	\caption{Examples of communities obtained from Boolean matrix factorization of the \erdos data.}
	\label{fig:erdos_BMF}
\end{figure}
\begin{figure}[!h]
	\centering
	\subfloat[]{
		\centering
		\includegraphics[width=0.472 \linewidth]{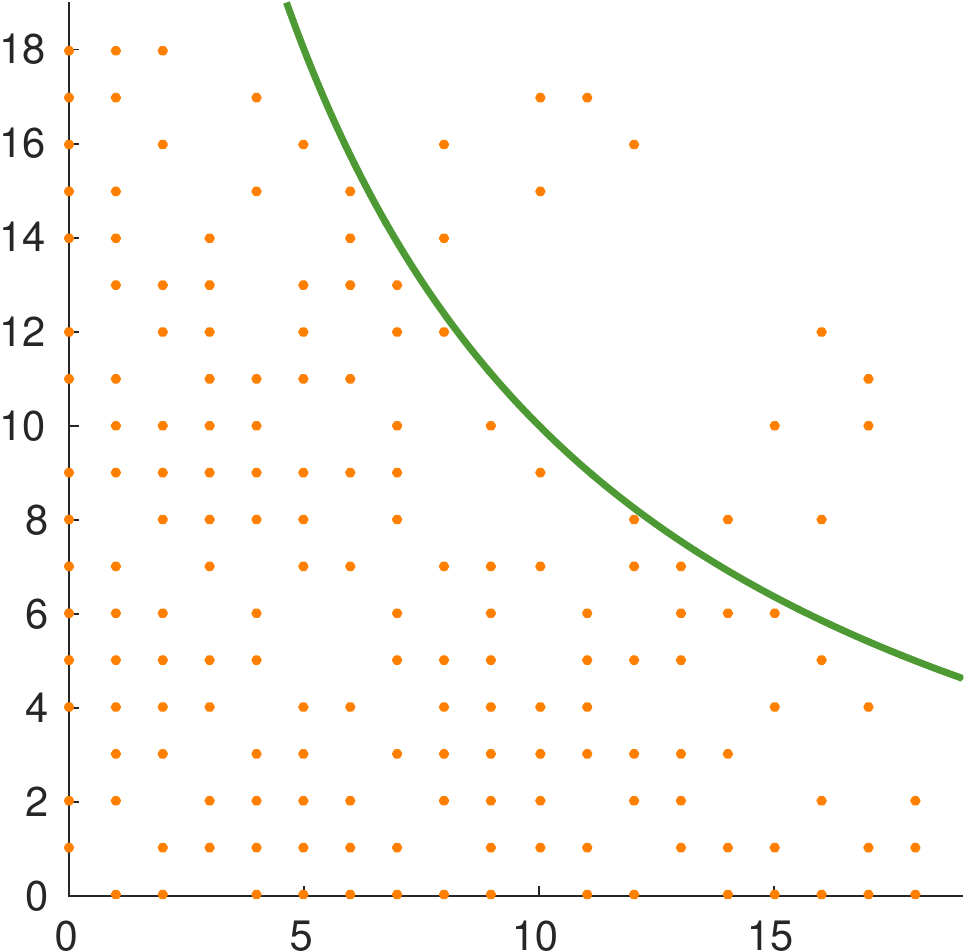}
		\label{fig:polbooks_BMF_1}
	}%
	\subfloat[]{
		\centering
		\includegraphics[width=0.472 \linewidth]{polbooks_BMF_2}
		\label{fig:polbooks_BMF_3}
	}
	\caption{Examples of communities obtained from Boolean matrix factorization of the \polbooks data.}
	\label{fig:polbooks_BMF}
\end{figure}
\begin{figure}[!h]
	\centering
			\subfloat[]{
			\centering
			\includegraphics[width=0.472 \linewidth]{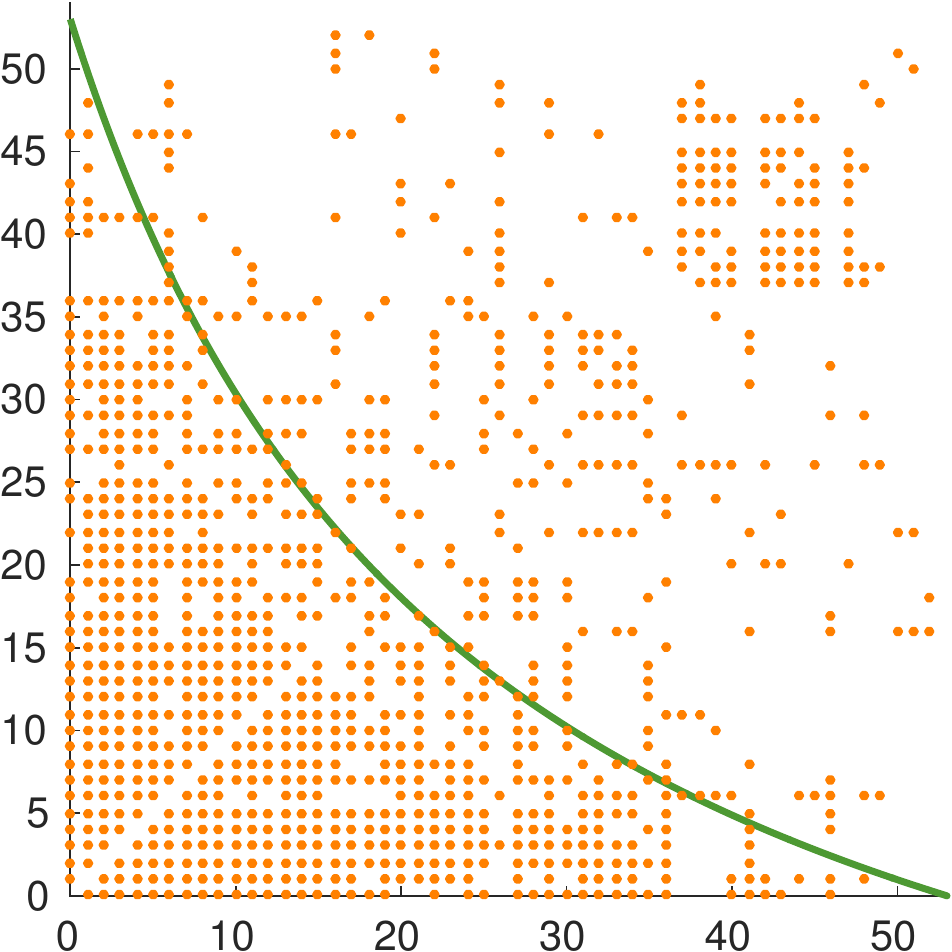}
			\label{fig:jazz_BMF_1}
		}%
		\subfloat[]{
			\centering
			\includegraphics[width=0.472 \linewidth]{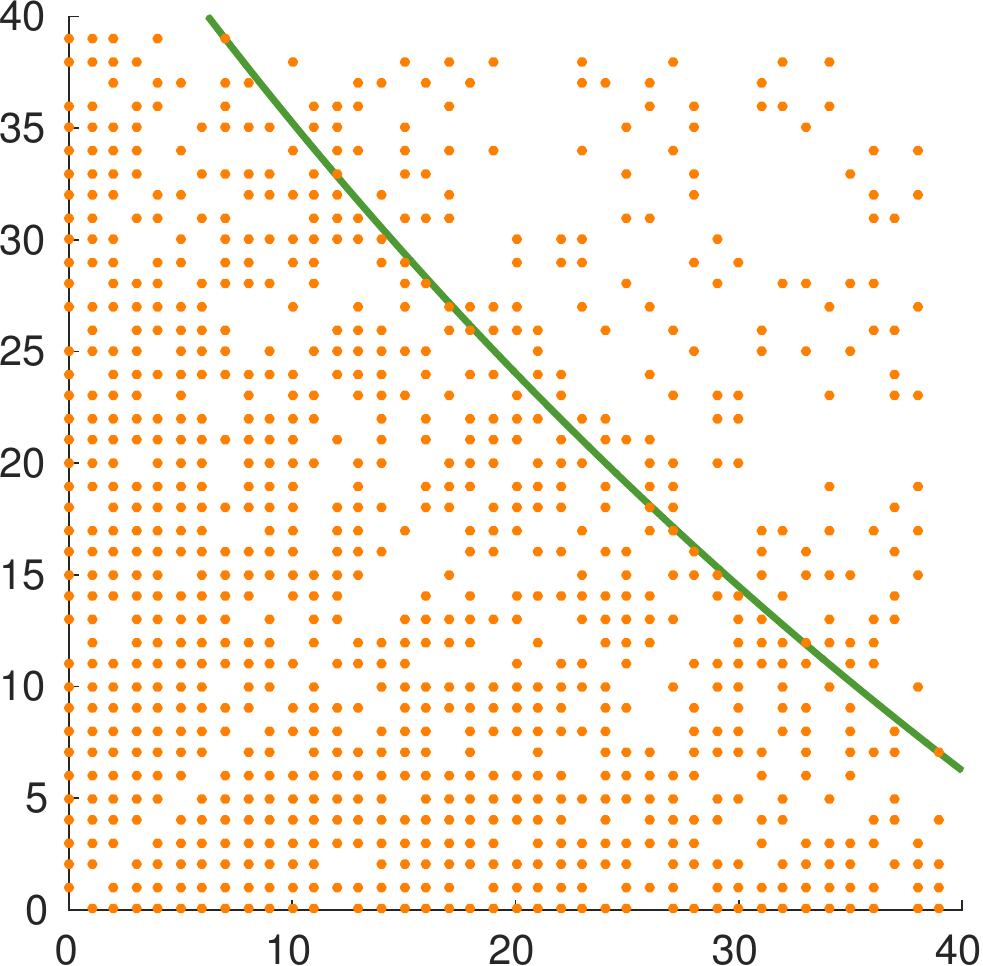}
			\label{fig:jazz_BMF_3}
		}
	\caption{Examples of communities obtained from Boolean matrix factorization of the \jazz data.}
	\label{fig:jazz_BMF}
\end{figure}
	

\end{document}